\pdfoutput=1

\documentclass[sigconf]{acmart}

\AtBeginDocument{%
  \providecommand\BibTeX{{%
    \normalfont B\kern-0.5em{\scshape i\kern-0.25em b}\kern-0.8em\TeX}}}

\setcopyright{acmcopyright}
\copyrightyear{2023}
\acmYear{2023}
\acmDOI{10.1145/1122445.1122456}

\acmConference[SIGMOD '23]{The ACM Special Interest Group on Management of Data}{June 18--23,
  2023}{Seattle, WA}
\acmPrice{15.00}
\acmISBN{978-1-4503-XXXX-X/18/06}

\acmSubmissionID{471}

\usepackage[scaled=0.78]{beramono}
\usepackage[ruled,vlined,linesnumbered]{algorithm2e}
\usepackage{graphicx}
\usepackage{caption}
\usepackage{subcaption}
\usepackage{amsfonts}
\usepackage{bm}
\usepackage{pifont}
\usepackage{booktabs,makecell,multirow}
\usepackage{siunitx}
\usepackage{multirow}
\usepackage{xspace}
\usepackage[capitalise]{cleveref}
\usepackage{tikz}
\usepackage{comment}
\usepackage{breqn} 
\usepackage{tabularx}
\usepackage{mathtools}
\usepackage{float}
\usepackage[shortlabels]{enumitem}
\setlist[enumerate]{leftmargin=1cm}
\setlist[itemize]{leftmargin=0.5cm}

\usepackage{breakurl}
\usepackage{listings}

\usepackage{pgfplots}
\usetikzlibrary{patterns}
\usetikzlibrary{calc,positioning,fit}
\usepgfplotslibrary{groupplots,units}

\usepackage{lipsum}
\usepackage{symbols}
\usepackage{thmtools}
\usepackage{thm-restate}

\setlist[enumerate]{nosep}
\setlist[enumerate]{nosep}
\setlist{nolistsep,leftmargin=20pt}
\crefname{section}{§}{§§}
\Crefname{section}{§}{§§}
\crefname{figure}{Fig}{Fig}
\Crefname{figure}{Fig}{Fig}

\usepackage{xpatch}

\newcommand{\name}{\textsf{\textsc{Wake}}\xspace}
\newcommand{\system}{\name}
\newcommand{\idf}{\textsf{edf}\xspace}
\newcommand{\mdf}{\idf}

\newcommand{\df}[1][]{\textsf{df}#1\xspace}
\newcommand{\cmark}{\ding{52}\xspace}
\newcommand{\xmark}{{\color{gray} \ding{55}}\xspace}
\newcommand{\tmark}{{\color{gray} \scalebox{.7}{\ding{115}}}\xspace}

\newcommand{\presto}{\textsf{Presto}\xspace}
\newcommand{\polars}{\textsf{Polars}\xspace}
\newcommand{\actianvector}{\textsf{Actian Vector}\xspace}
\newcommand{\vertica}{\textsf{Vertica}\xspace}
\newcommand{\postgres}{\textsf{Postgres}\xspace}

\newcommand{\progressivedb}{\textsf{ProgressiveDB}\xspace}
\newcommand{\wanderjoin}{\textsf{WanderJoin}\xspace}

\newcommand{\pending}[1]{{\color{red} #1}}

\newcommand{\todo}[1]{{\color{red} \footnotesize Note: #1}\xspace}

\newcommand{\ignore}[1]{{}}

\definecolor{BlueColor}{HTML}{0081a7}
\definecolor{RedColor}{HTML}{bb5555}
\definecolor{OrangeColor}{HTML}{bb5555}
\definecolor{YellowColor}{HTML}{ee9944}
\definecolor{GreenColor}{HTML}{38A528}
\definecolor{BrownColor}{HTML}{bc6c25}
\definecolor{PurpleColor}{HTML}{9b5de5}
\definecolor{BrownColor}{HTML}{8c564b}
\definecolor{PinkColor}{HTML}{e377c2}
\definecolor{GreyColor}{HTML}{d7d1d0}
\definecolor{BlackColor}{HTML}{d7d1d0}  

\definecolor{DarkGreenColor}{HTML}{264653}

\definecolor{wakecolor}{HTML}{bb5555}
\definecolor{polarcolor}{HTML}{38A528}
\definecolor{postgrescolor}{HTML}{bb5555}
\definecolor{prestocolor}{HTML}{38A528}
\definecolor{noscalecolor}{HTML}{38A528}
\definecolor{progressivedbcolor}{HTML}{38A528}
\definecolor{wanderjoincolor}{HTML}{0081a7}

\definecolor{wakemapecolor}{HTML}{0081a7}
\definecolor{wakeprecisioncolor}{HTML}{ee9944}
\definecolor{wakerecallcolor}{HTML}{ee9944}

\setcopyright{acmcopyright}
\pgfplotstableread{
column,size,count,mean,median,min,max
firstresult_best_mul,22,22.0,37.7,15.8,1.0,245.0
firstresult_polars_mul,22,20.0,27.4,13.8,1.0,125.5
firstresult_postgres_mul,22,22.0,661.6,303.3,4.1,5509.4
firstresult_presto_mul,22,22.0,118.2,104.9,5.1,499.8
overhead_best_mul,22,22.0,3.1,1.3,0.1,19.2
overhead_polars_mul,22,20.0,3.2,1.5,0.3,19.2
overhead_postgres_mul,22,22.0,0.6,0.1,0.0,7.3
overhead_presto_mul,22,22.0,1.1,0.3,0.1,6.5
overhead_best_perc,22,22.0,209.0,34.5,-92.5,1815.3
overhead_polars_perc,22,20.0,217.0,45.2,-65.5,1815.3
overhead_postgres_perc,22,22.0,-38.5,-86.9,-99.4,628.8
overhead_presto_perc,22,22.0,7.0,-67.5,-92.5,552.8
polars_memory_ratio,22,20.0,4.3,3.4,0.3,17.4
secondresult_best_mul,22,22.0,30.1,12.8,1.0,232.6
secondresult_polars_mul,22,20.0,19.8,11.8,1.0,77.2
secondresult_postgres_mul,22,22.0,531.8,238.3,2.6,4404.2
secondresult_presto_mul,22,22.0,91.6,78.3,3.8,307.4
}\LatencyMultiples

\def\medianfirsterror{$2.70\%$\xspace}

\def\medianmultipleonepercent{$3.17\times$\xspace}

\def\maxmultipleonepercent{$48.80\times$\xspace}

\def\medianmultipleonepercentola{$1.92\times$ \xspace{}}

\def\avgmemory{$4.3\times$\xspace} 
\def\maxmemory{$17.4\times$\xspace} 

\def\mediantimeactian{$4.93\times$\xspace}
\def\mediantimepolars{$11.8\times$\xspace}
\def\mediantimevertica{$21.81\times$\xspace}
\def\mediantimepostgres{$238.3\times$\xspace}
\def\mediantimepresto{$78.3\times$\xspace}
\def\mediantimebest{$4.93\times$\xspace}

\def\medianoverheadactian{$5.3\times$\xspace} %
\def\medianoverheadvertica{$0.8\times$\xspace} %
\def\medianoverheadpolars{$1.5\times$\xspace} 
\def\medianoverheadpostgres{$0.1\times$\xspace} 
\def\medianoverheadpresto{$0.3\times$\xspace} 
\def\medianoverheadbest{$1.3\times$\xspace} 

\author{
Nikhil Sheoran$^{+}\mbox{*}$, 
Supawit Chockchowwat$\mbox{*}$, 
Arav Chheda \\
Suwen Wang,
Riya Verma,
Yongjoo Park}
\email{nikhil.sheoran@databricks.com,  {supawit2,aravmc2,suwenw2,rverm2,yongjoo}@illinois.edu}
\affiliation{%
  \institution{Databricks$^{+}$ and University of Illinois Urbana-Champaign}
}
\renewcommand{\shortauthors}{Nikhil Sheoran, Supawit Chockchowwat, Arav Chheda, Suwen Wang, Riya Verma, Yongjoo Park}

\begin{document}
\title{A Step Toward Deep Online Aggregation}




\begin{abstract}

For exploratory data analysis,
    it is often desirable to know
        what answers you are likely to get
    \emph{before} actually obtaining those answers.
This can potentially be  achieved
    by designing systems to offer
        the estimates of a data operation result---say \textsf{op(data)}---earlier in the process
        based  on partial data processing.
Those estimates  continuously refine
        as more data is processed and finally converge to the exact answer.
Unfortunately,
    the existing techniques---called \emph{Online Aggregation} (OLA)---are limited to a single operation;
    that is, we \emph{cannot} obtain the estimates for \textsf{op(op(data))}
        or \textsf{op(...(op(data)))}.
If this \emph{Deep OLA} becomes possible, data analysts will be able 
    to explore data more interactively 
    using complex cascade operations.

In this work, we take a step toward \emph{Deep OLA}
    with \emph{evolving data frames} (\mdf),
        a novel data model
    to offer OLA
        for nested ops---\textsf{op(...(op(data)))}---by 
    representing an evolving structured data (with converging estimates) 
        that is \emph{closed} under set operations.
That is, \textsf{op(\mdf)}
    produces yet another \mdf;
    thus, we can freely apply successive operations to \mdf
        and obtain an OLA output for each op.
We evaluate its viability
    with \system, an \mdf-based OLA system,
        by examining against state-of-the-art OLA and non-OLA systems.
In our experiments on TPC-H dataset,
    \system produces its first estimates \mediantimebest faster (median)---with \medianoverheadbest median slowdown for exact answers---compared 
        to conventional systems.
Besides its generality,
    \system is also \medianmultipleonepercentola faster (median) than existing OLA systems
        in producing estimates of under 1\% relative errors.
\end{abstract}

\maketitle



\section{Introduction}
\label{sec:intro}


For ad-hoc data exploration,
the fastest way to gain insights
    would be to extract as much useful information as possible 
        from partial data processing
by computing intermediate estimates that
    highly resemble the (future) final answer,
while continuously refining the estimates until the entire data is processed.
Since the pioneering work by Hellerstein et al.~\cite{hellerstein1997online},
    this data processing paradigm---called \emph{Online Aggregation} (OLA)---has 
been studied in various directions
    to improve its generality and performance
    with new architectures~\cite{berg2019progressivedb,pansare2011online,wu2009distributed},
    novel join algorithms~\cite{li2016wander,haas1999ripple,urhan2000xjoin,mokbel2004hash,dittrich2002progressive,luo2002scalable,chen2010pr},
    support of subqueries~\cite{zeng2015g},
    specialized indexing~\cite{crotty2016case,kim2015rapid,park2016visualization,wu2010continuous},
    etc.

\begin{figure}[t]

\hspace*{-6mm}
\begin{tikzpicture}

\tikzset{
mynode/.style={
    font=\footnotesize,
    align=center,
    text width=0.32\linewidth
},
}

\begin{axis}[
    width=70mm,
    height=35mm,
    xmin=0,
    xmax=2,
    ymin=0,
    ymax=2,
    xtick={1,2},
    ytick={1,2},
    xticklabels={\textbf{One-time}, \textbf{Progressive/OLA}},
    yticklabels={
        \textbf{Template}\\ \textbf{Operations}, 
        \textbf{Cascade/Deep}\\ \textbf{Operations}
        },
    yticklabel style={align=center,yshift=-5mm,xshift=-2mm,rotate=0},
    xticklabel style={align=center,xshift=-14mm,yshift=0mm},
    tick label style={font=\footnotesize\sf},
    label style={font=\footnotesize\bf},
    grid=both,
    minor grid style={solid,gray},
]

\node[mynode] at (axis cs: 0.5, 1.5) { QuickR \cite{kandula2016quickr} };

\node[mynode] at (axis cs: 0.5, 0.5) { VerdictDB \cite{park2018verdictdb} };

\node[mynode] at (axis cs: 1.5, 0.5) 
    { RippleJoin \cite{haas1999ripple} 
        \\ WanderJoin \cite{li2016wander}
        \\ ProgressiveDB \cite{berg2019progressivedb} };

\node[mynode,font=\small] at (axis cs: 1.5, 1.5) {\textbf{This Work (\name)}};

\end{axis}

\end{tikzpicture}

\vspace{-6mm}
\caption{Our system \system enables OLA for deep operations}
\label{fig:position}
\vspace{-2mm}
\end{figure}

\def\thefootnote{+}\footnotetext{Work done while at University of Illinois at Urbana-Champaign.}\def\thefootnote{\arabic{footnote}}
\def\thefootnote{$^{}\mbox{*}$}\footnotetext{These authors contributed equally to this work.}\def\thefootnote{\arabic{footnote}}

\begin{table*}[t]
\begin{center}
\caption{Summary of existing work. \tmark/\xmark indicates limited/no support.}
\vspace{-3mm}
\footnotesize
\begin{tabular}{ l c c l l }
\toprule
 \textbf{System/Method} & \textbf{OLA?} & \textbf{Deep Query?} & \textbf{Novelty} 
    & \textbf{Weakness/Difference} \\ 
\midrule
 OLA (Hellerstein)~\cite{hellerstein1997online} & \cmark & \xmark & The first OLA proposal & Only for simple SQL with no joins/subqueries \\  
 RippleJoin~\cite{haas1999ripple,luo2002scalable} & \cmark & \tmark & Join algorithm for OLA & Exponential complexity for multiple joins \\
 WanderJoin~\cite{li2016wander} & \cmark & \tmark & Supports multiple joins & Requires indexes / May not produce exact answers \\
 G-OLA~\cite{zeng2015g}      & \cmark & \tmark & Supports filters with subqueries 
    & Some data need repetitive processing \\
 QuickR~\cite{kandula2016quickr}    & \xmark & \cmark & Pushes down sampling operators & Not OLA (each query answer is from a single sample) \\
 VerdictDB~\cite{park2018verdictdb}  & \xmark & \tmark & Platform-independent & Not OLA (each query answer is from a single sample) \\
 \textbf{Ours (\name)} & \cmark & \cmark & 
    Supports deeply nested operations & 
    May need more memory (\cref{sec:processing}); need to tune partition size (\cref{sec:ablation-partition}); no query optimizer \\
\bottomrule
\end{tabular}
\end{center}
\end{table*}

Unfortunately,
    the existing OLA has a common limitation,
        which makes it \emph{not} the first choice
    for today's data exploration.
That is,
    the existing OLA 
    \emph{precludes subsequent operations
        on previous OLA outputs.}
To illustrate this, 
    suppose a data analysis session\footnote{This example data analysis is a rewritten version of TPC-H query 18.}
        expressed in pandas-like methods~\cite{mckinney2011pandas} as follows:
\begin{lstlisting}[
    basicstyle=\ttfamily\small\linespread{0.9}\selectfont,
    frame=l,
    framesep=4.0mm,
    framexleftmargin=2.0mm,
    fillcolor=\color{GreyColor},
    rulecolor=\color{RedColor},
    numberstyle=\scriptsize,
    commentstyle=\color{GreenColor},
    xleftmargin=0.8cm,
    % columns=fullflexible,
    numbers=left,
    language=Python,
    stepnumber=1,
    morekeywords={join,sort}
    ]
lineitem = read_csv('...')
# item count for each order
order_qty = lineitem.sum(qty, by=orderkey)
# select only the large orders
lg_orders = order_qty.filter(sum_qty > 300)
# find the customers with biggest order sizes
lg_order_cust = lg_orders.join(orders).join(customer)
qty_per_cust = lg_order_cust.sum(sum_qty, by=name)
top_cust = qty_per_cust.sort(sum_qty, desc=True).limit(100)
\end{lstlisting}


The first output (\textsf{L3}) is
    aggregated/filtered again to find the top customers (\textsf{L5-L9}).
Existing OLA can incrementally compute the first output (i.e., \texttt{order\_qty}),
but it cannot be subsequently processed 
    for filter/join/sum in an OLA fashion
        until its final answer is obtained.
Specifically,
    the existing OLA has two limitations.
First, it treats every query independently
    without reasoning about
        how its output may be consumed by
    subsequent operations.
Second, it cannot handle arbitrarily deep queries;
    that is, even if we compose a (long) query for directly computing \texttt{avg\_order\_size},
        the aggregation over aggregation---with correct adjustments---cannot be produced
(\emph{note:} this problem is different from incremental view maintenance,
    which always produces the exact results).

In this work, 
    we tackle this limitation with
\emph{evolving data frames} (or \mdf),
    \emph{a new data/processing model designed to enable Deep Online Aggregation---the ability to
        apply subsequent operations to previous OLA outputs for another OLA output.}
For each operation, \mdf offers \emph{converging} estimates for the final answer---with 
    diminishing expected errors (\cref{sec:estimate:limit})---relying 
        on a common assumption that unseen data mimics the observed;
once the entire data is processed, each \mdf exactly matches the one that 
    can be obtained
        by conventional data systems.
To evaluate the viability of our approach,
    we implement \system\footnote{\system stands for \textbf{W}e \textbf{A}lready \textbf{K}now \textbf{E}nough.}, 
        an \mdf-based OLA system,
    and examine its performance
        against existing OLA systems (ProgressiveDB~\cite{berg2019progressivedb}, WanderJoin~\cite{li2016wander}) 
            as well as modern data systems/libraries (Presto~\cite{sethi2019presto}, PostgreSQL, Polars~\cite{polars}).
\emph{To our knowledge,
    no previous work has formally studied a data model
        for Deep OLA
    and has evaluated its efficiency for practical use cases
        with comprehensive experiments.}

\paragraph{Challenges}

Given a query (or an operation), existing OLA can be understood as 
    a process that 
        converts an input data into a series of intermediate/final results,
where notably, the input and the output are of different types,
    which is the fundamental reason that 
        OLA cannot be applied to the results of OLA.
Specifically, we observe the following challenges.
First, the existing model for structured data (which we call \emph{data frame})
    is insufficient for expressing progressively changing data frames
        which may contain approximate attribute values
            and their row counts may change.
%
Second, the existing set-oriented (or relational) operations are
    designed for a final data frame,
        not an approximate one;
simply applying regular operations to an evolving data frame may produce biased values
    because partial data must be regarded as a sample. 
%
Third, offering high performance is critical. 
Any OLA-driven extensions to the existing data model
    may incur overhead,
        which must be small enough to
    still deliver significantly more interactive computing 
        compared to conventional \emph{all-at-once} approaches.

\paragraph{Our Approach}

Our new data model, \mdf, is \emph{closed} under a  class of set operations;
that is, an \mdf
    expresses an evolving OLA output,
        which when transformed by a set operation,
    again produces yet another \mdf.
Specifically,
    \mdf has the following key characteristics.
First,
    each \mdf always converges to the exact/final answer once the entire data is processed.
Second, to ensure that an operation on an \mdf produces another \mdf,
    our set operations---expressed using \textsf{map}, \textsf{filter}, \textsf{join}, and \textsf{agg}---maintains 
    two unique properties inside each \mdf, i.e., 
        \emph{mutable attributes} and \emph{cardinality growth},
    which are key to producing accurate estimates (\cref{sec:motivation_properties_mdf}).
Third, our internal processing is designed to minimize redundant computation whenever possible.


\paragraph{Orthogonal Work}

OLA (including Deep OLA) can be understood as
    a mechanism that translates aggregation-involving queries
into an incremental computation logic,
    making it orthogonal to incremental 
        computation frameworks~\cite{murray2013naiad,mcsherry2013differential}
    and incremental view maintenance~\cite{ahmad2009dbtoaster,zeng2016iolap}.
OLA belongs to Approximate Query Processing, a broader class of query processing paradigm;
for example,
sampling-based AQP~\cite{park2018verdictdb,agarwal2013blinkdb,kandula2016quickr}
    produces a single approximate answer
(not a series of continuously refining answers like OLA).


\paragraph{Contributions}

This paper shares the following findings:
\begin{enumerate}
\item Our new data model, \emph{evolving data frames} (\mdf), enables successive OLA operations. (\cref{sec:model})
\item Our processing model, representing common set operations, can transform an \mdf into another \mdf
    (with correct properties) relying on our internal inference technique. (\cref{sec:processing} and \cref{sec:inference})
\item Our extended data model allows propagating confidence intervals through the pipeline. (\cref{sec:confidence})
\item \system's multi-thread implementation for OLA can offer high processing speed 
    with pipelined parallelism. (\cref{sec:system})
\item \system can produce an intermediate result \mediantimebest faster than
    the associated final answer, while \system incurs (only) \medianoverheadbest overhead compared to non-OLA.
    The median relative accuracy of the first intermediate answer for TPC-H datasets is \medianfirsterror,
        which converges quickly toward zero. (\cref{sec:exp})
\end{enumerate}
The technical contributions listed above
    appears after we present our motivation behind the design of evolving data frames (\cref{sec:motivation}).







\section{Motivation}
\label{sec:motivation}

We first describe a need for Deep OLA-specific data model (\cref{sec:motivation:type}).
To achieve Deep OLA, we present several cases our proposed framework must handle (\cref{sec:evolution:cases}),
    and discuss new unique properties (\cref{sec:motivation_properties_mdf}).

\input{figures/fig_motivation_cases}

\subsection{OLA Input/Output As Type}
\label{sec:motivation:type}

We argue that it is critical to 
    formalize the outputs of OLA as a \emph{type}.
For example,
    integers (e.g., -1, 0, 1, 2)
        are \emph{closed} under addition;
    thus, we can apply successive additions to the outputs of previous additions,
        e.g., (1 + 3) + 3 = 4 + 3,
    without being concerned about how the value 4 (= 1+3) is originally obtained.
Likewise, database relations (representing 2-D structured data)
    are \emph{closed} under relational operations (e.g., projection, aggregation, join).

In contrast, the existing OLA is designed to consume a relation as an input
    and outputs \emph{a series of} relations,
        each representing a converging estimate for the final answer 
            (with expectedly decreasing errors).
Thus, the input to and the output from OLA are of different types (i.e., relations are \emph{not closed} under OLA),
    which makes it non-trivial to apply successive OLA to the outputs of the previous OLA.
Theoretically,
    it might be possible to apply another OLA to each estimate relation;
    however, first, it is not straightforward how to interpret this,
        and second, the number of output estimates grows exponentially with the number of operations
        if we na\"ively apply OLA to each estimate.
This motivates us to introduce a new type that is \emph{closed} under OLA,
    which we call \emph{evolving data frame} (\mdf).

Unlike integers or relations, however,
    \mdf represents an \emph{evolving} object (not a \emph{state} of memory);
        thus, it is less obvious how we should understand/define
    the operations that map an evolving object to another evolving object.
We consider
    \mdf as an object consisting of
        a series of $K$ states,
    where each state contains
        a converging estimate for the final answer.
Also, we consider an operation on \mdf as a map
    from a set of states to another set of states
    by appropriately modifying the information inside each state (i.e., estimates and metadata)
        depending on the types and the parameters of the operation.

To formalize those states and operations on them,
    we start with a few example data frames
        that represent different types of
    transformations each data frame may go through
        during OLA (\cref{sec:evolution:cases}),
    based on which we will formalize \mdf in
the following sections (\cref{sec:model},\cref{sec:processing},\cref{sec:inference}).

\subsection{OLA Operations: Case Analysis}
\label{sec:evolution:cases}

With the example in \cref{sec:intro},
    we discuss how different operations (e.g., agg, filter, join, limit)
        may alter an input data frame into another form,
which serves as the basis of our data model in \cref{sec:model} and \cref{sec:processing}.

\paragraph{Order-preserving Local Operation}

Suppose an input data frame (e.g., \textsf{lineitem} table)---getting 
    read from csv file(s)---contain the raw data
        sorted/clustered on a key (e.g., \textsf{orderkey}).
In processing row-wise filters and maps,
    newly appearing rows in an input data frame
        do not affect the results of already processed rows.
From the example in \cref{sec:intro}, L1 (\textsf{read\_csv}), L3 (\textsf{group-by} on keys), 
    L5 (\textsf{filter}), 
    and L7 (\textsf{join}) belong to this category.
See \cref{fig:cases:a} for illustration.
Specifically,
    let \df be an input data frame consisting of two partitions,
        i.e., \df = [\df{$_1$}, \df{$_2$}{}],
    where ``[]'' indicates union/append.
For such a \emph{local operation} \textsf{op},
    \textsf{op(df) = [op(df$_1$), op(df$_2$)]};
thus, unlike other cases described shortly,
    computing \textsf{op(df$_2$)} is independent from \textsf{df$_1$},
which makes it possible to incrementally produce the output.
Likewise, inner/left join (e.g., joining \textsf{lineitem} with \textsf{orders})
    is also an order-preserving local operation
since \textsf{join(dfa, dfb) = [join(dfa$_1$, dfb), join(dfa$_2$, dfb)]};
    its physical plan  may opt for
        different join algorithms such as progressive-merge~\cite{dittrich2002progressive} or hash joins.


\paragraph{Shuffling Operation with Inference}

If an input data frame is aggregated by a non-key 
    attribute, e.g., \textsf{lg_order_cust.sum(sum_qty, by=name) in \cref{sec:intro}},
    we  need special considerations for three reasons.
First, already produced output (raw) aggregate values may change as we process more data from an input.
Second, raw aggregate values may need to be scaled appropriately to
    produce accurate---desirably unbiased---estimates.
Third, more rows (containing new grouping key values) may appear
    in the output
    as we process more data from an input,
        which need to be modeled quantitatively 
    for subsequent operations
        that will consume this output data frame.
From the example in \cref{sec:intro}, L8 (\textsf{group-by} on non-key attributes) belongs to this category.
\cref{fig:cases:b} depicts the data flows of this case:
    the newly appearing rows in the input may affect
    an already produced output.
Specifically, let \textsf{df = [df$_1$, df$_2$]}.
For a \emph{shuffling operation} \textsf{op},
        \textsf{op(df) = op(df$_1$) $\bigoplus$ op(df$_2$)},
    where $\bigoplus$ indicates a key-based merge,
        which can be expressed as A $\bigoplus$ B = agg(union(A, B), by=key).\footnote{
            Merge operations are applicable to sum-like (or \emph{mergeable}) operations,
                for which \emph{addition} can be defined. Accordingly, avg() needs to be computed
        by separately computing sum() and count(). One notably hard case is count-distinct,
    for which we maintain exact sets (not HLL-based sketches~\cite{flajolet2007hyperloglog}).}
%


The result of this merge must be scaled to produce accurate estimates 
    if \emph{more rows may appear in the input data frame} during OLA.
For example,
    if the input data frame represents a base table for which more data are being retrieved,
    the currently observed part(s) must be considered as a sample of the input data;
    thus, the raw sum values need to be scaled up in consideration of
        the ratio between the current  row count and the entire data size
        (which serves as a sampling ratio).
In contrast,
    if the input data frame is, for example, a result of an aggregation 
        (with a low-cardinality grouping attribute),
    we are unlikely to see (many) new rows in the input data frame;
thus, the currently observed set is the entire set, thereby not requiring additional scaling. 
While the individual aggregate values may be approximate,
    since they are already converging estimates of the final answer,
        the raw sum values (without additional scaling) are also converging estimates
            of the output (\cref{sec:estimate:limit}).


\paragraph{Shuffling Operation without Inference}

Operations like \textsf{order-by} and \textsf{limit}
    must consume the entire input,
        for which no special treatment can be applied to
    improve the quality of output.
In these cases,
    upon a change of input,
    the output simply needs to be recomputed in its entirety,
        which, unlike Cases 1 and 2,
    cause inevitably redundant computation.
From the example in \cref{sec:intro}, L9 (\textsf{order-by} and \textsf{limit}) belong to this category.
For large-scale aggregation, however,
    these Case 3 operations typically appear in the latter stages
        to limit/sort the result for user consumption (e.g., bar charts);
    thus, their overhead 
        is insignificant in the context of overall computations.
Nevertheless, if a user's intention is, for example, to sort the entire data
    and to persist its result on disk,
OLA frameworks (including ours)
    do not offer additional benefits.

\subsection{Required Properties}
\label{sec:motivation_properties_mdf}
The case analysis in \cref{sec:evolution:cases}
    reveals two types of changes:
changes to attribute values (e.g., as aggregating more input rows)
    and cardinality growth (e.g., filtering input rows as they appear).

\paragraph{Mutable Attributes}

Let a \emph{mutable attribute} be an attribute whose values may change whereas
    a \emph{constant attribute} be an attribute whose values \emph{never} change.
It is useful to distinguish mutable attributes from constant attributes
    because the input attribute types affect
        how we should (re-)compute the output.
For example,
    filtering on a constant attribute (e.g., \textsf{name like \textquotesingle\%east\%\textquotesingle})
can be processed incrementally (Case 1)
    whereas
        filtering on a mutable attribute (e.g., \textsf{sum_qty > 200})
    requires re-computation (Case 3).

\paragraph{Cardinality Growth}

As observed in Case 2,
    how many rows are likely to appear in an input data frame
    must be captured to properly estimate output aggregates.
To this end,
    we define \emph{cardinality growth}:
        the relationship between query progress and group cardinality (i.e. the number of rows belonging to an aggregate group).
    After studying a diverse family of cardinality growths, we can select the most fitting growth to predict the final aggregates.
For example,
    suppose we know that the group cardinality grows linearly with the query progress; 
then, if the query progress is at 25\%,
        we would expect to see $4\times$ rows in the final group cardinality.




\section{Data Model}
\label{sec:model}

This section describes evolving data frames (\mdf)
    from a user's perspective.
Specifically, we describe its data model (\cref{sec:mdf}),
    operations on it (\cref{sec:mdf:op}),
        and current limitations (\cref{sec:model:limit}).

\subsection{Evolving Data Frame}
\label{sec:mdf}

An evolving data frame (\mdf) represents 
    a progressively changing structured data (i.e, data frame)---with
        new rows appearing and/or changing attribute values---using
    the following formal definition:
\begin{lstlisting}[
    basicstyle=\ttfamily\small\linespread{0.9}\selectfont,
    fillcolor=\color{GreyColor},
    rulecolor=\color{RedColor},
    numberstyle=\scriptsize,
    commentstyle=\color{GreenColor},
    xleftmargin=0.8cm,
    % columns=fullflexible,
    % numbers=left,
    % language=Python,
    stepnumber=1,
    keywords={}
    ]
edf := t -> df  (0 <= t <= 1)
df := list((attr1, attr2, ..., attrM))
attr := constant_attr | mutable_attr
\end{lstlisting}    
where \textsf{(attr1, attr2, ..., attrM)} defines a schema.
One or more (constant) attributes 
        serve as the \emph{primary key} (or simply \emph{key})
            to uniquely identify tuples.
An \mdf's row count (i.e., the length of a list) may increase over time,
        and the values for \textsf{mutable_attr} may change.

\paragraph{Properties for Closure}
A valid \mdf must satisfy two properties, namely 2Cs.
\textbf{\textit{(1) Consistency}}: All the \texttt{df} associated with an \mdf has the same schema;
    that is, its list of attributes remains constant over $t$.
\textbf{\textit{(2) Convergence}}: The \texttt{df} associated with $t_2$
    is a more accurate estimator of the exact answer
        compared to the \texttt{df} associated with $t_1$ ($t_1 < t_2$)
    while the \texttt{df} at $t=1$ is the exact answer.
In other words,
    \emph{an operation on \mdf must produce an \mdf that ensures these two properties,
which guarantees that
    \mdf is \emph{closed} under those operations.}

\paragraph{Data Organization}
An \mdf's list is organized using one or more \emph{partitions},
    where each partition is (simply) a subset of the list stored/accessed together (e.g., on a storage device).
An \mdf may have a \emph{clustering key}, a list of attributes determining the placements of rows 
    among partitions;
    for example, if an \mdf's clustering key is \textsf{orderkey}, 
        a partition may include the rows with \textsf{orderkey} between 1 and 10;
    then, other partitions must not contain the rows with \textsf{orderkey=5}.
A clustering key may also be present for the \mdf{}s created as results of operations on other \mdf{}s,
    as we describe in \cref{sec:processing}.
\emph{A cluster of rows} refers to those rows present together in a partition.

\paragraph{Accessing Values}

As noted in \cref{sec:motivation}, to represent an evolving data frame,
    an \mdf maintains \emph{states},
    where each state expresses a converging estimate of the final answer
        with the latest state being the most accurate in expectation.
For example, 
    if an \mdf is for \textsf{lineitem.count(by=linestatus)},
        the count value in each row of the \mdf is an unbiased estimate of the final count value
            for the same group (i.e., they are equal in expectation).
The latest state is obtainable
    via \textsf{\mdf.get()}.
If \textsf{\mdf.is_final} is true,
    the latest state holds the final answer;
the \textsf{is_final} flag is set by the system as soon as the system finds there will be no more data
            to process (by receiving \textsf{eof}).
The final answer can be obtained by \textsf{\mdf.get_final()},
        which may block until processing the entire data (if not already processed).

\paragraph{Creating EDFs}

There are two ways to create \mdf{}s.
An \mdf can be created directly from a data source
    or an \mdf can be created as a result of the operation on another \mdf,
    as follows:
\begin{lstlisting}[
    basicstyle=\ttfamily\small\linespread{0.9}\selectfont,
    fillcolor=\color{GreyColor},
    rulecolor=\color{RedColor},
    numberstyle=\scriptsize,
    commentstyle=\color{GreenColor},
    xleftmargin=0.2cm,
    % columns=fullflexible,
    % numbers=left,
    % language=Python,
    stepnumber=1,
    keywords={}
    ]
read := data_source -> edf
edf_op := (edf, op) -> edf
op  := agg(attrs, by) | filter(predicate) 
     | map(function) | join(df, options)
agg := sum | count | avg | count_distinct | min | max 
     | var | stddev
\end{lstlisting}
where details of individual operations are described in \cref{sec:mdf:op}.
When creating an \mdf from a data source, a clustering key is obtained from metadata (\cref{sec:statistics}).
These operations and types of aggregations are sufficient to express all 22 TPC-H benchmark queries~\cite{tpch}.



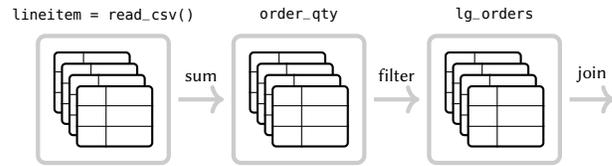
\begin{figure}[t]

\tikzset{
mytable/.pic={
    \node[
        draw=black,minimum height=8mm,minimum width=10mm,
        rounded corners=0.5mm,thick,fill=white,
    ] (-T) at (0,0) {};
    \draw[draw=black] 
        ($(-T.north west)!0.33!(-T.south west)$) -- ($(-T.north east)!0.33!(-T.south east)$);
    \draw[draw=black] 
        ($(-T.north west)!0.66!(-T.south west)$) -- ($(-T.north east)!0.66!(-T.south east)$);
    \draw[draw=black]
        ($(-T.north west)+(0.4,0)$) -- ($(-T.south west)+(0.4,0)$);
},
mylabel/.style={
    font=\tt\footnotesize,anchor=south
},
myborder/.style={
    draw=gray!40!white,ultra thick,rounded corners=1mm,
}
}

\begin{tikzpicture}

\def\s{2.6}

\def\x{0.1}
\def\y{0.15}
\draw pic (A1) at (0,0) {mytable};
\draw pic (A2) at (\x,-\y) {mytable};
\draw pic (A3) at (2*\x,-2*\y) {mytable};
\draw pic (A4) at (3*\x,-3*\y) {mytable};
\node[fit={($(A1-T.north west)+(-0.1,0.1)$) ($(A4-T.south east)+(0.1,-0.1)$)},myborder] (B1) {};
\node[mylabel] at ($(B1.north)+(0,0.05)$) {lineitem = read_csv()};

\draw pic (A1) at (\s,0) {mytable};
\draw pic at (\s+\x,-\y) {mytable};
\draw pic at (\s+2*\x,-2*\y) {mytable};
\draw pic (A4) at (\s+3*\x,-3*\y) {mytable};]
\node[fit={($(A1-T.north west)+(-0.1,0.1)$) ($(A4-T.south east)+(0.1,-0.1)$)},myborder] (B2) {};
\node[mylabel] at ($(B2.north)+(0,0.05)$) {order_qty};

\draw pic (A1) at (2*\s,0) {mytable};
\draw pic at (2*\s+\x,-\y) {mytable};
\draw pic at (2*\s+2*\x,-2*\y) {mytable};
\draw pic (A4) at (2*\s+3*\x,-3*\y) {mytable};
\node[fit={($(A1-T.north west)+(-0.1,0.1)$) ($(A4-T.south east)+(0.1,-0.1)$)},myborder] (B3) {};
\node[mylabel] at ($(B3.north)+(0,0.05)$) {lg_orders};

\draw[->,ultra thick,draw=gray!40!white] ($(B1.east)+(0.1,0)$) 
    -- node[above,midway,font=\sf\footnotesize,yshift=1mm] {sum} 
    ($(B2.west)+(-0.1,0)$);
\draw[->,ultra thick,draw=gray!40!white] ($(B2.east)+(0.1,0)$)
    -- node[above,midway,font=\sf\footnotesize,yshift=1mm] {filter} 
    ($(B3.west)+(-0.1,0)$);
\draw[->,ultra thick,draw=gray!40!white] ($(B3.east)+(0.1,0)$) 
    -- node[above,midway,font=\sf\footnotesize,yshift=1mm] {join} 
    ($(B3.east)+(0.7,0)$);

\end{tikzpicture}

\vspace{-2mm}
\caption{User-view of evolving data frames (\mdf) and operations on them.
    Each \mdf expresses one or more states.}
\label{fig:edf:user_view}
\vspace{-3mm}
\end{figure}

\subsection{Operations on Evolving Data Frame}
\label{sec:mdf:op}

Our system (\system) implements relational operations such as
    projection (map), join, selection (filter), and aggregation
        in a unique way---to maximize OLA opportunities---as follows.

\paragraph{Map} 

\textsf{\mdf.map()} resembles projection operations such as selecting a subset of attributes,
    creating derived attributes, etc.
What's unique to our \textsf{map()} is that the function in its argument
    is applied to one or more \emph{partitions}
        instead of each row.
Specifically, 
    let \textsf{\mdf = [p1, p2, ..., pK]} where \textsf{pi} is a partition of rows;
then, \textsf{\mdf.map(func)} creates another \mdf{}2
    such that \textsf{\mdf{}2 = [func([p1, p2]), func([p3, p4]), ..., func([pK-1, pK])]}.
That is, \textsf{func}
    maps a data frame to another data frame
where each input data frame is a set of partitions.
Here, the number of partitions passed to each \textsf{func} invocation---two in this example---is
    determined based on partition sizes.
            
There are two reasons behind this design.
First, this approach enables partition-specific (local) operations
    that are less trivial to express, e.g., 
        finding two most ordered items within each order where an order consists of multiple (item, quantity) tuples.
Expressing this using relational operations (in SQL)
    can involve less commonly used functions (e.g., \textsf{group\_concat}~\cite{mysql-group-concat}, \textsf{find\_in\_set}~\cite{mysql-find-in-set}).
Second,
    the approach easily enables efficient processing without additional logic
        for parallelizing/vectorizing row-wise functions.

\paragraph{Join} \textsf{\mdf.join(\mdf{}2, options)}
    joins \mdf with \mdf{}2
        as specified in its \textsf{options}, e.g., 
    method (inner/left) and join keys.
Depending on the join keys and clustering keys,
    \system uses a different join method (i.e., hash or progressive-merge~\cite{dittrich2002progressive}).
Specifically, if both \mdf and \mdf{}2 are clustered on their respective join keys,
    \system performs a merge join;
otherwise, \system performs a hash join with \mdf as the \emph{probe table}
    and \mdf{}2 as the \emph{build table} (used for creating a hash table).
If multiple joins are chained (e.g.,
    \textsf{\mdf.join(\mdf{}2).join(\mdf{}3)})
    and hash joins must be used,
        \system effectively performs the right-deep join
    by constructing hash tables in parallel for \mdf{}2 and \mdf{}3,
which is effective for star schema models~\cite{oracle-right-deep}.


\paragraph{Aggregate} 

\textsf{\mdf.agg(cols, by\_attr)}
    aggregates a group of rows (for each \texttt{by\_attr})
        where \textsf{agg} is one of the allowed aggregate functions.
For Deep OLA, 
    we treat aggregation specially
        because to generate accurate/unbiased estimates,
    the results of partial aggregation 
        may need adjustments
    in consideration of the ratio
        between an observed data frame size
            and the full data frame size,
    while the full data frame size may also be uncertain
        if, for example, the data we are aggregating is a result of another aggregation, 
    thereby requiring further inference.
\cref{sec:processing} and \cref{sec:inference} describe more on our inference logic.

\paragraph{Filter} 

\textsf{\mdf.filter(predicate)}
    resembles the selection operation in relational algebra (or the \textsf{where} clause in SQL);
that is, the operation produces another \mdf consisting only of 
    the rows satisfying the supplied predicate.
Like \textsf{\mdf.map(...)},
    the predicate is applied to one or more \emph{partitions} together.
In general,
    \textsf{filter()} can be understood an alias of \textsf{map()}
        that may produce an empty set as an output.
Specifically,
    for \textsf{\mdf2 = [func([p1, p2]), func([p3, p4]), ..., func([pK-1, pK])]},
        any of \textsf{func} may produce an empty set.

\vspace{2mm}
\noindent
We have described the four operations (i.e., map, join, aggregate, filter)
    from a user's perspective;
however, the internal processing may differ
    based on schemas/operations,
        which we describe in \cref{sec:processing}.

\subsection{Limitation}
\label{sec:model:limit}

There are cases where some operations must block
    (e.g., filtering/joining on mutable attributes)
        to produce correct results while minimizing redundant computations.
While our internal processing logic (\cref{sec:processing})
    can distinguish such cases,
        it may be less straightforward to end users
            especially when they are new to our framework.
To maximally exploit Deep OLA opportunities,
    more advanced users may carefully organize data and operations,
        which may be considered as
    \emph{skill} (like providing join hints in RDBMS);
however, one may argue that
    this means the system is not intelligent enough
        to automatically optimize user operations.
The scope of this work is 
    to construct the foundational building blocks for Deep OLA
without optimizing an end-to-end declarative query
    as performed by RDBMS with cost-based optimizers,
        which we leave as future work.



\section{Internal Processing}
\label{sec:processing}

\begin{figure}[t]

\begin{tikzpicture}

\tikzset{
myline/.style={
    draw=black,
    thick,
},
mylabel/.style={
    font=\footnotesize\sf,
    align=center,
},
}

\node[draw=black,thick,minimum width=50mm, minimum height=25mm] (A) at (0,0) {};

\draw[myline] ($(A.north)+(0,0.2)$) -- ($(A.south)+(0,-0.4)$);
\draw[myline] ($(A.north west)!0.33!(A.south west)+(-0.4,0)$) 
    -- ($(A.north east)!0.33!(A.south east)+(0.2,0)$);
\draw[myline] ($(A.north west)!0.66!(A.south west)+(-0.4,0)$) 
    -- ($(A.north east)!0.66!(A.south east)+(0.2,0)$);
    
\node[mylabel,anchor=north] (XA) at ($(A.south west)!0.25!(A.south east)+(0,-0.1)$) {schema has only\\constant\_attr};
\node[mylabel,anchor=north] (XB) at ($(A.south west)!0.75!(A.south east)+(0,-0.1)$) {schema includes\\mutable\_attr};

\node[mylabel,anchor=center] (LA) at ($(A.north west)!0.165!(A.south west)+(-1.0,0)$) {no growth\\($w = 0$)};
\node[mylabel,anchor=center] (LB) at ($(A.north west)!0.495!(A.south west)+(-1.0,0)$) {sub-linear\\($0 < w < 1$)};
\node[mylabel,anchor=center] (LC) at ($(A.north west)!0.83!(A.south west)+(-1.0,0)$) {linear\\($w = 1$)};

\node[mylabel,anchor=center] at ($(XA.north)!(LA.east)!(XA.south)$) {complete \mdf\\($t = 1$)};
\node[mylabel,anchor=center] at ($(XB.north)!(LA.east)!(XB.south)$) {agg by low-\\cardinality group};
\node[mylabel,anchor=center] at ($(XB.north)!(LB.east)!(XB.south)$) {agg by high-\\cardinality group};
\node[mylabel,anchor=center] at ($(XA.north)!(LC.east)!(XA.south)$) {read(base\_table)};

\end{tikzpicture}

\vspace{-4mm}
\caption{\mdf types categorized by degree of growth $w$ on Y-axis and attribute types on X-axis
    with examples in boxes.
    }
\label{fig:mdf:properties}
\vspace{-4mm}
\end{figure}

Depending on the types of operations,
    our system (\system) takes different approaches to update \idf.

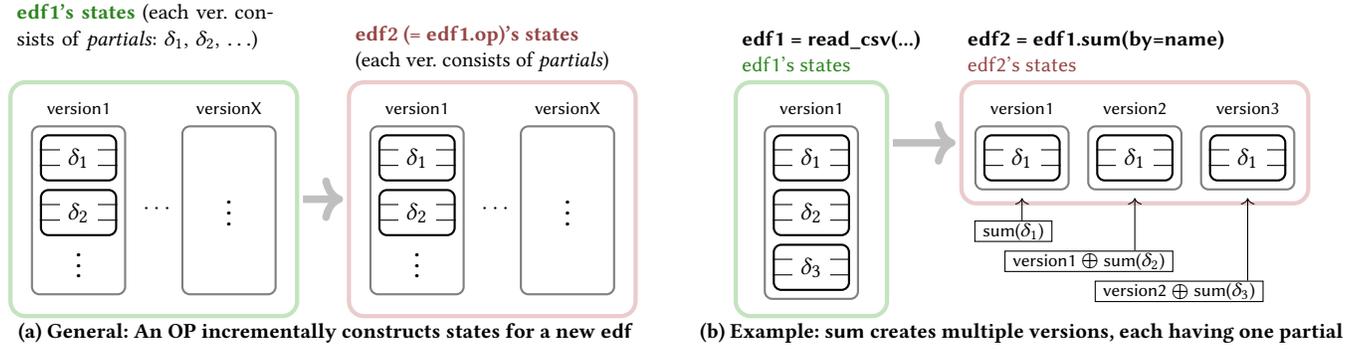
\begin{figure*}[t]

\tikzset{
mylabel/.style={
    font=\footnotesize\sf,
    align=center,
},
mytable/.pic={
    \node[
        draw=black,minimum height=6mm,minimum width=10mm,
        rounded corners=1mm,thick,
    ] (-T) at (0,0) {};
    \draw[draw=black] 
        ($(-T.north west)!0.33!(-T.south west)$) -- ($(-T.north east)!0.33!(-T.south east)$);
    \draw[draw=black] 
        ($(-T.north west)!0.66!(-T.south west)$) -- ($(-T.north east)!0.66!(-T.south east)$);
},
myannotation/.style={
    draw=black,anchor=north west,align=left,
    font=\sf\footnotesize,inner ysep=0.2mm,
}
}

\begin{subfigure}[b]{0.48\linewidth}
\centering
\begin{tikzpicture}

\draw pic (A1) at (0,0) {mytable};
\node[fill=white] at (A1-T.center) {$\delta_1$};
\draw[anchor=north] pic (A2) at ($(A1-T.south)+(0,-0.1)$) {mytable};
\node[fill=white] at (A2-T.center) {$\delta_2$};
\node[anchor=north] (A3) at ($(A2-T.south)+(0,0.1)$) {\textbf{$\vdots$}};
\node[mylabel,anchor=south] at ($(A1-T.north)+(0,0.15)$) {version1};

\coordinate (E) at ($(A3.south west)!(A2-T.south east)!(A3.south east)$);
\node[fit={(A1-T.north west) (E)},draw=gray,thick,rounded corners=1mm] (K1) {};


\node[anchor=west,font=\small] at ($(K1.east)+(0.1,0)$) {\textbf{$\cdots$}};

\draw[,opacity=0] pic (A1) at (2.0,0) {mytable};
\draw[anchor=north,opacity=0] pic (A2) at ($(A1-T.south)+(0,-0.1)$) {mytable};
\node[anchor=north,opacity=0] (A3) at ($(A2-T.south)+(0,0.1)$) {\textbf{$\vdots$}};
\node[] at ($(A2-T.center)+(0,0.1)$) {\textbf{$\vdots$}};
\node[mylabel,anchor=south] at ($(A1-T.north)+(0,0.15)$) {versionX};

\coordinate (E) at ($(A3.south west)!(A2-T.south east)!(A3.south east)$);
\node[fit={(A1-T.north west) (E)},draw=gray,thick,rounded corners=1mm] {};

\node[fit={(-0.8,0.9) (2.8,-2.0)},draw=GreenColor!30!white,ultra thick,rounded corners=2mm] (B1) {};
\node[font=\small,anchor=south west,text=GreenColor!80!black,align=left,text width=35mm] 
    at ($(B1.north west)+(0,0)$) 
        {\textbf{\mdf{}1's states} {\color{black} (each ver.~consists of \emph{partials}: $\delta_1$, $\delta_2$, $\ldots$) }};

\def\s{4.5}
\draw pic (A1) at (\s,0) {mytable};
\node[fill=white] at (A1-T.center) {$\delta_1$};
\draw[anchor=north] pic (A2) at ($(A1-T.south)+(0,-0.1)$) {mytable};
\node[fill=white] at (A2-T.center) {$\delta_2$};
\node[anchor=north] (A3) at ($(A2-T.south)+(0,0.1)$) {\textbf{$\vdots$}};
\node[mylabel,anchor=south] at ($(A1-T.north)+(0,0.15)$) {version1};

\coordinate (E) at ($(A3.south west)!(A2-T.south east)!(A3.south east)$);
\node[fit={(A1-T.north west) (E)},draw=gray,thick,rounded corners=1mm] (K1) {};


\node[anchor=west,font=\small] at ($(K1.east)+(0.1,0)$) {\textbf{$\cdots$}};

\draw[,opacity=0] pic (A1) at (\s+2.0,0) {mytable};
\draw[anchor=north,opacity=0] pic (A2) at ($(A1-T.south)+(0,-0.1)$) {mytable};
\node[anchor=north,opacity=0] (A3) at ($(A2-T.south)+(0,0.1)$) {\textbf{$\vdots$}};
\node[] at ($(A2-T.center)+(0,0.1)$) {\textbf{$\vdots$}};
\node[mylabel,anchor=south] at ($(A1-T.north)+(0,0.15)$) {versionX};

\coordinate (E) at ($(A3.south west)!(A2-T.south east)!(A3.south east)$);
\node[fit={(A1-T.north west) (E)},draw=gray,thick,rounded corners=1mm] {};

\node[fit={(\s-0.8,0.9) (\s+2.8,-2.0)},draw=OrangeColor!30!white,ultra thick,rounded corners=2mm] (B2) {};
\node[font=\small,anchor=south west,text=OrangeColor!80!black,align=left] 
    at ($(B2.north west)+(0,0)$) 
        {\textbf{\mdf{}2 (= \mdf{}1.op)'s states}\\ {\color{black} (each ver.~consists of \emph{partials}) }};

\draw[->,line width=1mm,draw=gray!50!white] 
    ($(B1.east)+(0.05,0)$) -- ($(B2.west)+(-0.05,0)$);

\end{tikzpicture}
\vspace{-2mm}
\caption{General: An OP incrementally constructs states for a new \mdf}
\end{subfigure}
\hfill
\begin{subfigure}[b]{0.48\linewidth}
\centering
\begin{tikzpicture}

\draw pic (A1) at (0,0) {mytable};
\node[fill=white] at (A1-T.center) {$\delta_1$};
\draw[anchor=north] pic (A2) at ($(A1-T.south)+(0,-0.1)$) {mytable};
\node[fill=white] at (A2-T.center) {$\delta_2$};
\draw[anchor=north] pic (A3) at ($(A2-T.south)+(0,-0.1)$) {mytable};
\node[fill=white] at (A3-T.center) {$\delta_3$};
\node[mylabel,anchor=south] at ($(A1-T.north)+(0,0.15)$) {version1};

\coordinate (E) at (A3-T.south east);
\node[fit={(A1-T.north west) (E)},draw=gray,thick,rounded corners=1mm] (K1) {};

\node[fit={(-0.9,0.9) (0.9,-2.0)},draw=GreenColor!30!white,ultra thick,rounded corners=2mm] (B1) {};
\node[font=\small\sf,anchor=south west,text=GreenColor!80!black,align=left] 
    at ($(B1.north west)+(0,0)$) 
        {{\color{black} \textbf{\mdf{}1 = read_csv(...)} } \\ \mdf{}1's states };
        
\def\s{3}
\draw pic (A1) at (\s-0.2,0) {mytable};
\node[fill=white] at (A1-T.center) {$\delta_1$};
\node[mylabel,anchor=south] at ($(A1-T.north)+(0,0.15)$) {version1};
\coordinate (E) at ($(A1-T.south east)+(0,0)$);
\node[fit={(A1-T.north west) (E)},draw=gray,thick,rounded corners=1mm] (K1) {};

\draw[] pic (A1) at (\s+1.3,0) {mytable};
\node[fill=white] at (A1-T.center) {$\delta_1$};
\coordinate (E) at ($(A1-T.south east)+(0,0)$);
\node[mylabel,anchor=south] at ($(A1-T.north)+(0,0.15)$) {version2};
\node[fit={(A1-T.north west) (E)},draw=gray,thick,rounded corners=1mm] (K2) {};

\draw[] pic (A1) at (\s+2.8,0) {mytable};
\node[fill=white] at (A1-T.center) {$\delta_1$};
\coordinate (E) at ($(A1-T.south east)+(0,0)$);
\node[mylabel,anchor=south] at ($(A1-T.north)+(0,0.15)$) {version3};
\node[fit={(A1-T.north west) (E)},draw=gray,thick,rounded corners=1mm] (K3) {};

\node[fit={(\s-0.9,0.9) (\s+3.5,-0.5)},draw=OrangeColor!30!white,ultra thick,rounded corners=2mm] (B2) {};
\node[font=\small\sf,anchor=south west,text=OrangeColor!80!black,align=left] 
    at ($(B2.north west)+(0,0)$) 
        {{\color{black} \textbf{\mdf{}2 = \mdf{}1.sum(by=name)} } \\ \mdf{}2's states };

\coordinate (O) at ($(B1.north east)!(B2.west)!(B1.south east)$);
\draw[->,line width=1mm,draw=gray!50!white] 
    ($(O)+(0.05,0)$) -- ($(B2.west)+(-0.05,0)$);
    
\node[myannotation] (P1)
    at ($(B2.south west)+(0.2,-0.2)$) 
    {sum($\delta_1$)};
\node[myannotation] (P2)
    at ($(B2.south west)+(0.6,-0.6)$) 
    {version1 $\bigoplus$ sum($\delta_2$)};
\node[myannotation] (P3)
    at ($(B2.south west)+(1.8,-1.0)$) 
    {version2 $\bigoplus$ sum($\delta_3$)};

\draw[->,black] ($(P1.north west)!(K1.south)!(P1.north east)$) -- ($(K1.south)+(0,-0.1)$);
\draw[->,black] ($(P2.north west)!(K2.south)!(P2.north east)$) -- ($(K2.south)+(0,-0.1)$);
\draw[->,black] ($(P3.north west)!(K3.south)!(P3.north east)$) -- ($(K3.south)+(0,-0.1)$);

\end{tikzpicture}

\vspace{-2mm}
\caption{Example: \textsf{sum} creates multiple versions,
    each having one partial}
\end{subfigure}

\vspace{-3mm}
\caption{States-based representation of \mdf and operations on them
    (i.e., from \mdf{}1's extrinsic states to \mdf{}2's intrinsic states).
Each \mdf (conceptually) defines a new state
    each time a partial is added to the latest version or
        a new version is created.
An operation on \mdf
    creates a new \mdf (and states for it)
        with minimal computational redundancy.
    }
\label{fig:mdf:states}
\end{figure*}

\subsection{Properties of Evolving DataFrame}
\label{sec:processing:properties}

In addition to the schema described in \cref{sec:mdf},
    each \mdf maintains two additional properties, namely
        \emph{progress} and \emph{growth},
        to characterize its evolution quantitatively.

\paragraph{Progress}

\emph{Progress} ($0 \le t \le 1$)
        is the ratio between the number of (original) \emph{input} tuples that have been read/processed thus far
            and the \emph{total} number of the (input) tuples that must be processed
    to obtain the final answer;
        the total tuple count comes from  metadata (\cref{sec:statistics}).
For example,
    if a base table consists of ten equal-sized partitions
        and we have read/processed only one of them,
    $t$ is 1/10 (= 0.1).
On the other hand, if the entire data (e.g., ten out of ten partitions) is read/processed,
        $t$ is 1.
If $t$ is 1, \textsf{\mdf.is_final=True}.

\paragraph{Growth}

\emph{Growth} describes the growth of the current tuple count to forecast the final tuple count. \system compactly models the growth as a monomial $c t^w$ using past observations. \cref{fig:mdf:properties} gives some examples with different $w$ values. Growth captures the local tuple count, while progress $t$ captures the query input ratio (between the current and the future).
For instance,
    if we are computing an average (without grouping attributes),
        the output tuple count will always be one (unless empty); thus, $w = 0$ and $c = 1$.
    On the other hand, $t < 1$ if we are still reading/processing input data.



\paragraph{Examples}

These variables---$(c, w)$ and $t$---are
    more closely related if, for example,
        an \mdf represents a base table; then, $w$ is equal to $1$ and $c$ is equal to the input size,
    because in this case, the output of this \mdf (or the data this \mdf represents)
        exactly matches the amount of input data retrieved from a data source (e.g., CSV files in a directory).
In other cases, however,
        $w$ may be less than $1$, suggesting sub-linear growth.
For instance,
    if an \mdf
        represents the result of aggregation
    with log-cardinality grouping attributes---\textsf{lineitem.count(by=linestatus)}---the 
number of output rows is less likely to increase 
    (while its aggregate values may change);
    thus, we have $t = 1$.
%
\cref{fig:mdf:properties} classifies the types of 
    \mdf properties based on the degree of growth ($w$)
        and attribute types (constant/mutable).
Its cells list a few examples
    that would result in \mdf{}s with such properties.
For example, if \textsf{\mdf = read(base\_table)},
    its schema consists only of constant attributes and its output size
        grows linearly with input data ($w = 1$).
Another example is
    an \mdf representing the result of aggregation with high-cardinality grouping attributes
        (\textsf{students.count(by=first\_name)}).
If so, attribute values may change,
    and also, new grouping keys may appear (each time a new first name is encountered).
Accordingly, its schema includes mutable attributes, and $w$
    is between $0$ (no growth) and $1$ (linear growth).

\subsection{State Representation}

Internally, an \mdf
    represents an evolving data frame
        with discrete \emph{states}.
There are two types of states: \emph{intrinsic states} and \emph{extrinsic states}.
Extrinsic states express converging/unbiased estimates; accordingly,
        they are consumed by downstream \mdf{}s or other applications,
whereas
    intrinsic states are used to incrementally maintain computed values
        prior to adjustments and/or estimations.

\paragraph{Examples}

Suppose we are counting the number of students by their home states.
Let \mdf{}1 represent the dataset we are reading;
we have read one out of ten equal-sized partitions,
    the first partition contains 2 students from IL and 1 student from MI.
The intrinsic states $\alpha_1$ of \mdf{}1 becomes
    \textsf{[[(id1, IL), (id2, IL), (id3, MI)]]}.
For \mdf{}1, its extrinsic states $\beta_1$ is identical to $\alpha_1$ because \mdf{}1---representing
    tuples from a base table---requires no adjustments.
Let \mdf{}2 represent \textsf{\mdf{}1.count(by=state)};
    its intrinsic states $\alpha_2$ becomes
        \textsf{[[(IL, 2), (MI, 1)]]}.
To express unbiased estimates,
    \mdf{}2's extrinsic states $\beta_2$ is scaled accordingly
        under the assumption that the unobserved (nine) partitions
            have the same distribution as the observed (first) partition;
thus, $\beta_2$ becomes \textsf{[(IL, 20), (MI, 10)]}.

We read one more partition (thus, we have read two partitions);
    the second partition contains 1 student from IL and 1 student from MI.
$\alpha_1 = \beta_1$ becomes
    \textsf{[[(id1, IL), (id2, IL), (id3, MI)], [(id4, IL), (id5, MI)]]} (note that
        the newly added tuples are in a separate list).
To (incrementally) update $\alpha_2$,
    we first aggregate the second list of $\beta_1$,
        temporarily obtaining \textsf{[(IL, 1), (MI, 1)]},
    which is merged into $\alpha_2$ using key-based sum ($\bigoplus$), 
        as described in \cref{sec:evolution:cases},
    finally obtaining \textsf{$\alpha_2$ = [(IL, 3), (MI, 2)]}.
To obtain unbiased estimates from $\alpha_2$,
    we scale individual aggregate values considering the ratio
    between currently processed tuples and the total tuple count (i.e., 2:10),
thereby obtaining \mdf{}2's extrinsic states \textsf{$\beta_2$ = [(IL, 15), (MI, 10)]}.

Note that we have taken two different approaches in updating intrinsic states
    depending on \mdf{}s.
For \mdf{}1, we have inserted new tuples, creating a longer list for $\alpha_1$;
    in contrast, for \mdf{}2, we have replaced the old set of aggregate values 
        with another set of aggregate values.
We systematically distinguish these cases---incremental or complete updates---as follows.

\paragraph{Intrinsic States}

To enable both incremental and complete updates,
    an \mdf's states are organized using \emph{versions} and \emph{partials}
        (a partial is a subset of rows inside each version),
    as shown in \cref{fig:mdf:states}.
Creating a new version means a complete refresh
    while appending partial(s) to each version (of an \mdf)
        means incremental updates.

For example,
    suppose an \mdf---representing \textsf{(first\_name, count)} statistics of 
a class---has a version $\alpha^{(1)}$ and
        the version currently contains one partial,
    where the partial has one tuple
    (e.g., \textsf{[(mike, 4)]}).
We can incrementally update the version by appending another partial (e.g., \textsf{[(sarah, 2)]});
    then, the version $\alpha^{(1)}$ represents two tuples \textsf{[(mike, 4), (sarah, 2)]},
        namely a union of the two partials.

\newcommand{\balpha}{\bm{\alpha}}

Specifically,
    intrinsic states $\balpha$
        is a two-dimensional structure (\cref{fig:mdf:states}),
    consisting of one or more versions ($\alpha^{(1)}, \alpha^{(2)}, \ldots, \alpha^{(v)}$),
        where each version $\alpha^{(i)}$
    contains one or more partials ($\delta_1, \ldots, \delta_p$).
The partials are exclusive from one another with respect to their key;
    that is, the partials \emph{partition} each version,
        which is ensured by each \mdf during operations (\cref{sec:states:updates}).
To obtain the latest intrinsic state,
    we can union all the partials in the latest version ($\balpha^{(v)}$).

\paragraph{Extrinsic States}

Extrinsic states are introduced to distinguish
    (external) estimate values
        from (internal) raw values.
In many operations such as map/filter/join,
    the extrinsic states are simply an alias of intrinsic states
        since those operations do not need any special adjustments to obtain unbiased estimates.
Extrinsic states are required primarily for aggregate operations.
    
There are two types of adjustments.
    The first is 
        when aggregation is \emph{non-mergeable} (e.g., count-distinct),
    requiring different pre-aggregate representations.
Let \textsf{\mdf{}2 = \mdf{}1.count_distinct(name)},
    where \mdf{}1's intrinsic states $\alpha_1$ consist of two partials $\delta_1$ and $\delta_1$.
To incrementally compute count-distinct (i.e., first using $\delta_1$ and then to update it using $\delta_1$),
    it is insufficient to have the number of unique values appearing in $\delta_1$
    because \textsf{count\_distinct($\delta_1$) + count\_distinct($\delta_2$)}
        is \emph{not} equal to \textsf{count\_distinct($\delta_1 \cup \delta_2$)};
    we need to record all the individual unique values in $\delta_1$
        to properly examine if the tuples in $\delta_2$ overlap with any of the values in $\delta_1$.
In this case, the intrinsic states must include a set of unique values,
    which then can be used to incrementally compute count-distinct values
        (finally appearing in extrinsic states).

The second type is when aggregate values are expected to increase/change
    if we observe more tuples in the input:
        currently observed tuples should be treated as a sample.
One example is a sum, as we have already described.
That is, by treating the current raw summation as the ones from a sample,
    unbiased estimates can be obtained in consideration of
            the ratio between the current input cardinality
                and the projected final input cardinality.
This scaling mechanism (called growth-based scaling)
    is described in \cref{sec:inference}.

\subsection{Operation: State Transformation}
\label{sec:states:updates}

An operation in \textsf{\mdf{}2 = \mdf{}1.op(...)}
    is a state transformation process
        from \mdf{}1's extrinsic states
    to \mdf{}2's intrinsic states
(which can then be used to produce \mdf{}2's extrinsic states
    with optional scaling as described above).
    In this section, we describe
        how to transform
a version of extrinsic states $\beta_1 = [\delta_1, \ldots, \delta_p]$
    \emph{incrementally} to a version of intrinsic states $\alpha_2$
    for each operation.

\begin{table}[t]
\begin{center}
\caption{State transformation for each \mdf operation. GBI: growth-based inference.}
\label{tab:merge_rules}
\vspace{-2mm}
\footnotesize
\begin{tabular}{ l | l l l }
\toprule
\textbf{\mdf op} & \textbf{intrinsic repr.} 
    & \textbf{merge ($\bigoplus$)} & \textbf{int. $\rightarrow$ ext.} \\ 
\midrule
map    & mapped tuples   & union      & identity \\
join   & joined tuples   & union      & identity \\
filter & filtered tuples & union      & identity \\
count  & count by key    & sum by key & GBI  \\
sum    & sum by key      & sum by key & GBI  \\
avg    & sum/count by key  & sum/sum by key & GBI  \\
count\_distinct  & count by key    & sum by key & GBI  \\
min    & min by key      & min by key & GBI  \\
max    & max by key      & max by key & GBI  \\
var    & var/sum/count by key      & avg/sum/sum by key & GBI  \\
stddev    & var/sum/count by key      & avg/sum/sum by key & GBI  \\
\bottomrule
\end{tabular}
\end{center}
\vspace{-4mm}
\end{table}

\paragraph{Merge}

To incrementally construct $\alpha_2$ with respect to \textsf{op}
    (when provided $\delta_1, \ldots, \delta_p$ at a time),
we exploit the fact that
    there exists 
        a combination of an intrinsic state representation
            and a \emph{merge} operation ($\bigoplus$) 
    that can satisfy
        \textsf{op([$\delta_1, \ldots, \delta_p$]) = op($\delta_1$) $\bigoplus$ \ldots $\bigoplus$ op($\delta_p$)}.
That is, given $\delta_1$, we can first compute op($\delta_1$);
    then, given $\delta_2$, we update the result by merging op($\delta_2$) into the previous result;
this update operation continues for each partial.

For example, suppose we are computing \textsf{avg([$\delta_1$, $\delta_2$, $\delta_3$])},
    or more specifically, average salary for each state in the United States.
To incrementally compute average,
    we first compute \textsf{(count, sum_salary)} for each state from $\delta_1$,
        which is stored as an intrinsic state.
Given the next partial ($\delta_2$), we (again) compute 
    \textsf{(count, sum_salary)} for each state from $\delta_2$,
then add these aggregates into the earlier results for each state,
    which is equal to directly computing \textsf{(count, sum_salary)}
        from a union of $\delta_1$ and $\delta_2$.
Note that
    for each \textsf{op},
        these intrinsic state representations and merge operations differ,
    which we summarize in \cref{tab:merge_rules}.

\paragraph{Primary Key}

As described in \cref{sec:mdf},
    one or more constant attributes serve as a primary key
to uniquely identify tuples of an \mdf.
Accordingly, our transformation
    always defines a primary key for a newly created \mdf.
\textsf{map/filter/join} retains the same key as the input \mdf.
Upon \textsf{agg}, grouping attributes becomes the key of a new \mdf.

\paragraph{Clustering Key}

A clustering key determines the physical ordering of an \mdf's tuples.
The clustering key changes by aggregation
    if the aggregation's grouping attributes
        are not the clustering key itself.

\paragraph{Other Properties}

Besides attribute types,
    a new \mdf must also maintain two internal properties: progress and growth.
Since progress is a ratio defined using the original input tuples,
    every operation simply propagates the progress value to the next \mdf
        without any modifications.
In contrast,
    growth is newly calculated as part of an operation
        (each time a new partial or a version is consumed)
to accurately estimate the number of tuples that will newly appear in the future.
We discuss this logic in \cref{sec:inference}.





\subsection{Base Table Statistics}
\label{sec:statistics}

The \mdf that represents a base table (by reading data from CSV, Apache Parquet, or others)
    must be provided with (1) a list of file names,
        (2) the number of tuples in each file,
    and (3) attributes with primary/clustering keys corresponding to the tables that are being read.
This is all the metadata that \system requires from the underlying data, without requiring any other statistics.
This metadata information
    is used for computing \textit{progress} ($t$).

\subsection{Closure of \mdf Properties}
\label{sec:estimate:consistent}
\label{sec:estimate:limit}

\system's internal processing is designed to ensure 2Cs (\cref{sec:mdf})
    required to ensure the validity of all \mdf{}s through processing.
Specifically, we satisfy (1) consistency and (2) convergence, as follows.

\paragraph{Consistency}
Every operation in \cref{sec:mdf:op} is a function mapping input \mdf{}(s) to an output \mdf with a fixed schema. 
Because the source of \mdf (\texttt{read} operation) always generates an \mdf with a fixed schema, 
all intermediate and final \mdf{}s have the same schema.



\paragraph{Convergence}
There are two types of convergence: (1) mutable attributes become more accurate, and
    (2) the key set (e.g., group-by attributes) converges.
\quad \textit{Attribute convergence}:
First of all, all the attribute values produced by \system are \emph{convergent}.
    That is, let $\tilde{x}_n$ be an attribute value of an \mdf associated with a certain key 
        after processing up to the $n$-th tuple,
        whereas the exact value---the value we obtain after processing the entire data---is $x$.
Then, two properties hold: first, $\E[|\tilde{x}_n - x|] \le \E[|\tilde{x}_{n'} - x|]$ for $n \le n'$;
    and second, $\tilde{x}_N = x$ where $N$ is the total tuple count.
While desirable, the latter property ($\tilde{x}_N = x$) is often not ensured
    by some existing OLA systems that rely on statistical simulations~\cite{li2016wander}.
Moreover,
for mean-like aggregates (e.g., count, sum, avg, stddev, var), we produce unbiased estimates; that is,
    $\E[\tilde{x}_n - x] = 0$.
For other aggregates (e.g., count-distinct, extreme order statistics like min/max),
    we produce reasonably accurate estimates adopting well-known estimation techniques in the literature~\cite{Haas1995Distinct,Vaart1998AsympStats}.
\textit{Key-set convergence}: 
In approximate computing, a major source of non-existing keys is insufficient
    samples from the input data~\cite{chaudhuri2017approximate}.
Nevertheless, under our framework, the key set converges to the true set
    because our operations are designed to produce the exact answers when the entire input data is observed.




\section{Aggregate Inference}
\label{sec:inference}



Given an \mdf's intrinsic states, aggregate inference produces its extrinsic states.
There are two challenges.
First, group sizes (e.g., the number of students from a certain state)
    may grow in a non-linear way as more input data are processed.
Second, the number of groups may also increase over time (i.e., the number of states).
Third, different types of aggregations often require different estimation mechanisms.
To tackle these challenges, 
    our overall inference logic (\cref{sec:inference-decompose}) decomposes into two parts: 
        cardinality estimator (\cref{sec:inference-cardinality}) and 
        aggregate estimators (\cref{sec:inference-aggregate}).

\subsection{Problem Decomposition}
\label{sec:inference-decompose}

\system formulates aggregate inference as an unbiased estimation problem. 
Using intrinsic states up until current progress $0 \leq t \leq 1$, aggregate inference aims to find per-cell unbiased estimators at final progress $T = 1$. 
Suppose the data frame has $m(t)$ groups, $X_i(t)$ denotes the $i$-th group cardinality (i.e. the number of tuples that have been aggregated into the $i$-th group) at progress $t$.
Although many aggregate attributes may be present, aggregate inference focuses on each attribute at a time, 
    referring to the aggregate values of the $i$-th group as $Y_i(t)$. 
Because unobserved partitions are unknown to \system, 
    $m(t)$, $X_i$, and $Y_i$ are stochastic processes over ``time'' $t$. 
We write the observed group cardinalities and aggregate values until progress $t$ in lower cases: 
    $x_{i, :t}$ and $y_{i, :t}$ respectively. 
The desired unbiased estimator $\yh_{i, :t}$ is the one such that:
\begin{equation}
    \yh_{i, :t} = \E \left[ Y_i(T) \; | \; x_{i,:t}, y_{i,:t} \right]
\end{equation}

Many known aggregate estimators rely on the current count $x_{i, t}$ 
    and the final count $x_{i, T}$; 
however, the latter is not known at the current time. 
Instead, \system computes an unbiased estimator of final group cardinality $\xh_{i, :t}$ 
    from group cardinalities so far $x_{i, :t}$ 
        described in \cref{sec:inference-cardinality}.
\begin{equation}
    \xh_{i, :t} = \E \left[ X_i(T) \; | \; x_{i, :t} \right]
\end{equation}

Using this estimator as well as the current cardinality and aggregate value, 
    \system then estimates the final aggregate value at $T$ 
        by aggregate-aware estimators $f$ in \cref{sec:inference-aggregate}.
\begin{equation}
    \yh_{i, :t} = \E \left[ Y_i(T) \; | \; x_{i,:t}, y_{i,:t} \right] = f(y_{i,:t},  x_{i, :t}, \xh_{i, :t})
\end{equation}

Therefore, \system first estimates $\xh_{i, :t}$ for all $i = 1, \dots, m(t)$, 
    and then estimates $\yh_{i, :t}$. 
In a case of many aggregate attributes, \system reuses $\xh_{i, :t}$ 
    to estimate each aggregate separately by applying the corresponding aggregate estimator. 
Finally, it collects all aggregate estimations into the output data frame filled with extrinsic states.

\subsection{Cardinality Estimator}
\label{sec:inference-cardinality}

\system models group cardinalities after \emph{monomials} with a shared power, 
    $\E[X_i(t)] \propto t^w$. 
The underlying reasoning is as follows. 
    \system assumes that the number of samples and the number of groups follow two hidden monomials, 
        $\E[n(t)] \propto t^{u}$ and $\E[m(t)] \propto t^{v}$, respectively. 
Then, average group cardinality is $\frac{1}{m(t)} \sum_{i=1}^{m(t)} X_i(t) = \frac{1}{m(t)} n(t)$ 
    whose expectation is proportional to $t^{u - v}$,
so 
$w = u - v$. 
This 
modeling 
captures many scenarios in Deep OLA. 
For example, if the input data frame is a table reader, 
    then \system would expect the sample to grow linearly ($u = 1$). 
If the input is behind a cross join of two tables, 
    then \system would expect a quadratic growth ($u = 2$). 
Filtering would then affect the coefficient corresponding to its selectivity. 
On the other hand, if the group key is the same as the clustering key, 
    \system would see the number of groups grow linearly ($v = 1$) 
        as it consumes more partitions. 
A low-cardinality group key would result in a constant ($v = 0$) 
    while a higher-cardinality one would generate something in between ($0 < v \leq 1$). 

Furthermore, this model simplifies its estimation logic. 
    In fact, \system does not need to estimate $\E[n(t)]$ nor $\E[m(t)]$, 
        but only $\E[X_i(t)]$. 
\system estimates 
final group cardinalities in two steps. 
    First, it fits $w$ to the dataset consisting of average group cardinalities $\overline{x}_t = \frac{1}{m_t} \sum_{i=1}^{m_t} x_{i, t}$ for all observed $t$. 
Specifically, it fits the power $w$ as well as the coefficient $b$ in a logarithmic-transformed linear regression: $\E[\log \overline{x}_t] = \log b + w \log t$. 
\system implements a streaming linear regression 
    with $O(1)$ time/space complexities per observation. 
Finally, \system fits each group's coefficient 
    in $\E[X_i(t) | x_{i, :t}] = x_{i, t} = c_i t^w$ 
        and predicts the final group cardinality 
with $T = 1$:
\begin{equation} \label{eq:cardinality-estimate}
    \xh_{i, :t} = \E[X_i(T) | x_{i, :t}] = (x_{i, t} / t^w)\;  T^w = x_{i, t} \; / \; t^w
\end{equation}

\begin{figure*}[t]
\tikzset{
execnode/.style={
    circle,draw=black,minimum width=6mm,
    font=\scriptsize\sf,
    ultra thick,
},
mytable/.pic={
    \node[
        draw=black,minimum height=4mm,minimum width=8mm,
        rounded corners=1mm,thick,
    ] (-T) at (0,0) {};
    \draw[draw=black] 
        ($(-T.north west)!0.33!(-T.south west)$) -- ($(-T.north east)!0.33!(-T.south east)$);
    \draw[draw=black] 
        ($(-T.north west)!0.66!(-T.south west)$) -- ($(-T.north east)!0.66!(-T.south east)$);
},
mylabel/.style={
    font=\scriptsize\sf,anchor=north,inner xsep=0.5mm,
    minimum width=8.5mm,align=left,
},
myarrow/.style={
    draw=GreenColor,ultra thick,->
},
myannotation/.style={
    draw=black,font=\scriptsize,color=RedColor,
    rounded corners=1mm,
},
}

\begin{subfigure}[b]{0.4\linewidth}
\centering
\begin{tikzpicture}
\node[] at (0,0) {
\begin{lstlisting}[
    basicstyle=\ttfamily\footnotesize\linespread{1.1}\selectfont,
    % frame=l,
    % framesep=4.0mm,
    % framexleftmargin=2.0mm,
    fillcolor=\color{GreyColor},
    % rulecolor=\color{RedColor},
    % numberstyle=\scriptsize,
    commentstyle=\color{GreenColor},
    % xleftmargin=0.8cm,
    % columns=fullflexible,
    % numbers=left,
    language=Python,
    % stepnumber=1,
    morekeywords={join,sort,limit}
    ]
lineitem = read_csv('...')          # LI
# item count for each order
order_qty = lineitem.sum(qty, by=orderkey)     # OQ
# select only the large orders
lg_orders = order_qty.filter(sum_qty > 300)    # LO
# find the customers with biggest order sizes
lg_order_cust = lg_orders.join(orders) \       # OO
                  .join(customer)              # OC
# select top-100 customers
qty_per_cust = lg_order_cust.sum(sum_qty, by=name) # C
top_cust = qty_per_cust.sort(sum_qty, desc=True) \
                .limit(100)              # TC
\end{lstlisting}
};
\end{tikzpicture}

\end{subfigure}
\hfill
\begin{subfigure}[b]{0.55\linewidth}
\centering
\begin{tikzpicture}

\def\xs{1.4}
\def\ys{0.0}
\def\ls{0.15}

\node[execnode] (L) at (0,0) {LI};
\node[mylabel,anchor=center] at ($(L.north)+(0,\ls)$) {read_csv};

\draw pic (T1) at ($(L.south)+(0,-0.6)$)  {mytable};
\draw pic (T2) at ($(T1-T.south)+(0,-0.3)$) {mytable};
\draw pic (T3) at ($(T2-T.south)+(0,-0.3)$) {mytable};
\node[mylabel,anchor=north west] (L1) at ($(T3-T.south west)+(0,-0.05)$) 
    {lineitem\\(key: orderkey,\\ linenum)};
\node[
    fit={(T1-T.north west) (L1.south east)},
    draw=gray!40!white,
    thick, rounded corners=1mm, 
    ] (B1) {};
\draw[draw=gray!40!white,ultra thick] (L) -- ($(B1.north west)!(L)!(B1.north east)$);

\node[execnode] (O) at ($(L)+(\xs,\ys)$) {OQ};
\node[mylabel,anchor=center] at ($(O.north)+(0,\ls)$) {sum};
\node[mylabel,anchor=center] at ($(O.south)+(0,-\ls)$) {(key: orderkey)};

\node[execnode] (LO) at ($(O)+(\xs,\ys)$) {LO};
\node[mylabel,anchor=center] at ($(LO.north)+(0,\ls)$) {filter};

\node[execnode] (OO) at ($(LO)+(\xs,\ys)$) {OO};
\node[mylabel,anchor=center] at ($(OO.north)+(0,\ls)$) {(merge-) join};

\node[execnode] (OD) at ($(OO)+(-0.4*\xs,-0.8)$) {OD};
\draw pic (T5) at ($(OD.south)+(0,-0.55)$)  {mytable};
\draw pic (T52) at ($(T5-T.south)+(0,-0.3)$) {mytable};
\node[mylabel,align=right,anchor=north east] (L2) at ($(T52-T.south east)+(0,-0.05)$) 
    {orders\\(key: orderkey)};
\node[
    fit={(T5-T.north east) (L2.south west)},
    draw=gray!40!white,
    thick, rounded corners=1mm, 
    ] (B2) {};
\draw[draw=gray!40!white,ultra thick] (OD) -- ($(B2.north west)!(OD)!(B2.north east)$);

\node[execnode] (OC) at ($(OO)+(\xs,\ys)$) {OC};
\node[mylabel,anchor=center] at ($(OC.north)+(0,\ls)$) {(hash-) join};

\node[execnode] (CS) at ($(OC)+(-0.4*\xs,-0.8)$) {CS};
\draw pic (T6) at ($(CS.south)+(0,-0.6)$)  {mytable};
\node[mylabel,inner xsep=0,anchor=north west] (L3) at ($(T6-T.south west)+(0,-0.05)$) 
    {customer\\(key: custkey)};
\node[
    fit={(T6-T.north west) (L3.south east)},
    draw=gray!40!white,
    thick, rounded corners=1mm, 
    ] (B3) {};
\draw[draw=gray!40!white,ultra thick] (CS) -- ($(B3.north west)!(CS)!(B3.north east)$);

\node[myannotation,anchor=north west,text width=28mm] (A3) 
    at ($(OC.south)+(0.4,-0.6)$) 
    {\textbf{note:} uses the right table (i.e., \textbf{customers}) for a build table.
    The key is still orderkey.};
\coordinate (A3N) at ($(A3.north)+(-1.2,0)$);
\draw[draw=RedColor] ($(A3N)+(0.05,0)$) to[out=90,in=-30] (OC);

\node[execnode] (C) at ($(OC)+(\xs,\ys)$) {C};
\node[mylabel,anchor=center] at ($(C.north)+(0,\ls)$) {sum};
\node[mylabel,anchor=center] at ($(C.south)+(0,-\ls)$) {(key: name)};

\node[execnode] (T) at ($(C)+(\xs,\ys)$) {TC};
\node[mylabel,anchor=center] at ($(T.north)+(0,\ls)$) {sort, limit};


\node[myannotation,anchor=north west,text width=20mm,fill=white] 
    (A1) at ($(B1.north east)+(-0.4,-0.15)$) 
    {\textbf{note:} incrementally read \textsf{\textbf{lineitem}} clustered on \textsf{orderkey} };
\draw[draw=RedColor] (L) -- ($(A1.north west)+(0.2, 0)$);

\node[myannotation,anchor=south east,text width=18mm,fill=white] 
    (A5) at ($(B2.south west)+(0.6,0.6)$) 
    {\textbf{note:} incrementally read \textsf{\textbf{orders}} clustered on \textsf{orderkey} };
\draw[draw=RedColor] (OD) -- ($(A5.north east)+(-0.2, 0)$);

\node[myannotation,anchor=south west,rounded corners=0mm,color=GreenColor,thick,text=black,
    font=\footnotesize] 
    at ($(L.north west)+(-0.4,0.5)$) 
    {\textbf{note:} right arrows indicate message queues between 
        nodes; each node runs on a separate thread};

\draw[myarrow] (L) -- (O);
\draw[myarrow] (O) -- (LO);
\draw[myarrow] (LO) -- (OO);
\draw[myarrow] (OO) -- (OC);
\draw[myarrow] (OC) -- (C);
\draw[myarrow] (C) -- (T);
\draw[myarrow] (OD) -- (OO);
\draw[myarrow] (CS) -- (OC);

\end{tikzpicture}

\end{subfigure}

\vspace{-2mm}
\caption{Example data operations with \mdf (left) and 
    \system's internal representation to process them in parallel (right)
    }
\label{fig:implementation-query-example}
\vspace{-2mm}
\end{figure*}
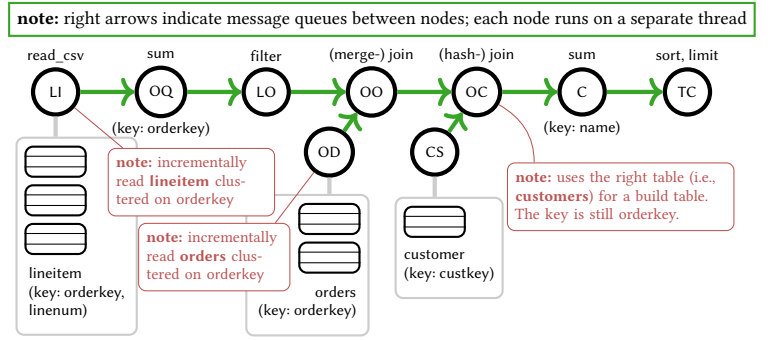

\subsection{Aggregate Estimators}
\label{sec:inference-aggregate}

\system selects the aggregate estimator $f$ from the following set of estimators depending on the aggregation type. This set can be expanded together with existing estimators.

\paragraph{Count} Use the estimated cardinality: $f_{\text{count}}(y_{i,:t}, x_{i,:t}, \xh_{i, :t}) = \xh_{i, :t}$.

\paragraph{Sum} Scale the summation: $f_{\text{sum}}(y_{i,:t}, x_{i,:t}, \xh_{i, :t}) = \frac{y_{i,t}}{x_{i,t}} \xh_{i, :t}$.

\paragraph{Weighted Avg} Weighted averages (e.g., average, variance, standard deviation) are special cases of summation. Because of our choice of estimators, average estimators reduce to the identity function. Let $y'_{i,:t}$ be the weighted summation, $y''_{i,:t}$ be the summation of weights and $y_{i,t} = y'_{i,t} / y''_{i,t}$ be the weighted average:
\begin{equation} \label{eq:scale-weighted-avg}
    f_{\text{avg}}((y'_{i,:t}, y''_{i,:t}), x_{i,:t}, \xh_{i, :t}) = \left( \frac{y'_{i,t}}{x_{i,t}} \xh_{i, :t} \right) / \left( \frac{y''_{i,t}}{x_{i,t}} \xh_{i, :t} \right) = y_{i,t}
\end{equation}
\

\paragraph{Count Distinct} 

\system adopts a finite-population method-of-moment estimator~\cite{Haas1995Distinct} 
    (in \cref{sec:processing:properties}, denoted as $\hat{D}_{MM1}$). 
For brevity in this subsection, 
    let us focus on the $i$-th group and 
        shorten the notations of current group cardinality $x = x_{i,t}$, 
            final estimated group cardinality $X = \xh_{i, :t}$, 
    and current group count distinct $y = y_{i, t}$. 
\system computes $f_{cd}(y_{i,:t}, x_{i,:t}, \xh_{i, :t}) = Y$ where $Y$ satisfies \cref{eq:mm1-main}.
\begin{equation} \label{eq:mm1-main}
    y_{i, t} = Y ( 1 - h(\xh_{i, :t} / Y) )
\end{equation}

$h(z)$ is defined below. To solve the equation, \system runs Newton-Raphson iterations until convergence with a tolerance and at most a finite number of steps. Each iteration involves evaluating the numerical approximation of gamma and digamma functions.
\begin{equation}
    h(z) = \frac{\Gamma(\xh_{i, :t} - z + 1) \, \Gamma(\xh_{i, :t} - x_{i,t} + 1)}{\Gamma(\xh_{i, :t} - x_{i,t} - z + 1) \, \Gamma(\xh_{i, :t} + 1)}
\end{equation}

\paragraph{Order Statistics} 
Order statistics include min, max, median, quantiles, and $k$-th smallest/largest values. 
Currently, \system simply outputs the latest value: 
    $f_{\text{order}}(y_{i,:t}, x_{i,:t}, \xh_{i, :t}) = y_{i,t}$ 
        which provides a fairly accurate estimate for large $\xh_{i,:t}$ at no computation cost.

\subsection{Correctness}
\label{sec:inference-analysis}

Given observations $(y_{i,:t}, x_{i,:t})$ up until current progress $t$, \cref{lemma:unbiased-count} and \cref{lemma:unbiased-agg} together show that \system's aggregate inference is unbiased under some conditions.

\begin{restatable}[Unbiased Count]{lemma}{lemmaunbiasedcount}
\label{lemma:unbiased-count}
    $\xh_{i, :t} = \E[X_i(T) | x_{i, :t}] = \frac{x_{i,t}}{t^w}$ is unbiased, if A) $w$ is unbiased and B) all operations produce a monomial or transform a monomially growing input(s) into a monomially growing output with respect to progress $t$.
\end{restatable}

\ignore{
\begin{proof}
    To justify A), $w$ is unbiased if a sufficient condition holds: $\overline{X}_t = \frac{1}{m} \sum_{i=1}^{m(t)} X_i(t) = c_i t^w e^\eps$ where $e^\eps$ is a multiplicative error with $\E[\eps] = 0$ and independence to observations. Notice that $\log \overline{X}_t = \log c_i + w \log t + \eps$. By unbiasedness of the ordinary least square estimator, the estimated $w$ coincides with the expectation of $w$.

    To justify B), we show that all operations in \cref{sec:model} can satisfy the property possibly with some additional conditions depending on the operation. First, a data source produces monomial rows by implementation since it computes progress $t$ from the number of rows read divided by the final number of rows. Mapping trivially preserves the cardinality. Joins and filters preserve monomial relationship if each tuple has uniform selectivity and the selectivity is independent of the progress. Lastly, aggregation preserves monomial relationship if the groupby key is either uniformly distributed across progress or is a clustered key.
    
    If B) holds, $\E[X_i(t)] = c_i t^w$ for all \mdf including the input to the aggregation. Because \system has already observed $x_{i, t}$ at $t$, $\E[X_i(t) | x_{i,:t}] = \E[X_i(t) | x_{i,t}] = x_{i, t}$, implying $c_i = x_{i,t} / t^w$ and leading to \system's group cardinality estimator when evaluated at $t = T = 1$.
\end{proof}
}

\begin{restatable}[Unbiased Aggregation]{lemma}{lemmaunbiasedagg}
\label{lemma:unbiased-agg}
    Given the unbiased group cardinality estimate $\xh_{i, :t} = \E[X_i(T) | x_{i, :t}]$, \system's aggregate estimators produce unbiased estimates, possibly with additional conditions depending on aggregation type: $\E \left[ Y_i(T) \; | \; x_{i,:t}, y_{i,:t} \right] = f(y_{i,:t},  x_{i, :t}, \xh_{i, :t})$.
\end{restatable}

\ignore{
\begin{proof}
    This is true for $f_{count}$ because $\xh_{i, :t}$ is given.
    
    Assuming the terms $\{U_{i,j}\}_{j=1}^{X_i(t)}$ in summation are stationary with a constant $\E[U_{i,j}] = \E[U_i]$, $f_{sum}$ is unbiased. Let $Y_i(t) = \sum_{j=1}^{X_i(t)} U_{i,j}$, we can use stationarity to factor $\E[Y_i(T) \; | \; x_{i,:t}, y_{i,:t}] = \E[X_i(T) \; | \; x_{i,:t}] E[U_i \; | \; x_{i,:t}, y_{i,:t}]$. By expanding the conditional expectation, $E[U_i \; | \; x_{i,:t}, y_{i,:t}] = \frac{y_{i,t}}{x_{i,t}}$ which equates the final result to $f_{sum}$.
    
    Assuming that elements and weights are independent, weighted average estimate $f_{avg}$ is unbiased by expectation linearity.
    
    Assuming equal-frequency assumption (i.e. samples are uniformly distributed into $Y$ distinct values), $f_{cd}$ is an unbiased estimator as described in~\cite{Haas1995Distinct}. Additionally, the same work also derives an estimator for skewed frequency, left for future improvements and experiments.

    The $q$-th sample quantile is asymptotically normal around the $q$-th population quantile with variance $q(1-q) / (X_i(t) p_q^2)$ where $p_q$ is the density of the quantile value studied in section 21.2 of~\cite{Vaart1998AsympStats}. Assuming sufficiently large sample size $x_{i,t} = \omega(1/p_q^2)$, which is typical in OLA due to its data size, the sample quantile is unbiased with small variance.
\end{proof}
}

\noindent
Please find the proofs in our extended manuscript~\cite{deepola-tech}.



\subsection{Alternatives}
\label{sec:inference-alternative}

This section lists some of the alternative design choices we have considered but do not fit well with the broader picture of \system.

\paragraph{Probabilistic Cardinality Estimator} One could model the distribution $X_i(T) | x_{i,:t}$ (instead of the expectation $\E[X_i(T) | x_{i,:t}]$ in \system) to express confidence. However, the evaluation would require computing the marginal expectation which may be expensive for many aggregate estimators. Moreover, which distribution to use is an open question to be investigated further.



\paragraph{Other Cardinality Function Families} Different families of polynomials are attractive alternatives; 
    however, one needs to know the set of orders \emph{a priori} to efficiently fit their coefficients. 
Mixing exponential and logarithm could improve the accuracy in some cases 
    but would also be more difficult to estimate. 
In contrast, affine functions are simple with many well-known estimation algorithms, but they are only restricted to a specific growth pattern. Ultimately, monomial is the simplest and cheap to fit (in logarithmic scale) 
    yet provides a wide range of growth curves.


\paragraph{Order Statistics under Finite Population} 

Given PDF/CDF,
    there exists a density function of the $k$-th order statistic~\cite{david2004order}. 
Given that,
    we could evaluate the expectation at $\xh_{i, :t}$ numerically 
        to acquire an unbiased estimate. 
A similar analysis is possible for discrete variables as well. 
However, this method has a prohibitive computational cost in general to reconstruct PDF/CDF, let alone evaluating the expectation. 
For example, if the reconstruction uses kernel density estimation (KDE)~\cite{rosenblatt1956remarksKDE,parzen1962estimationKDE} and empirical distribution function (eCDF), 
    it would require $O(n(t))$ time and space to compute the density and hold all samples. 
Such a cost does not scale well and may straggle OLA progress with minimal accuracy gain.

\section{Confidence Interval for Deep OLA}
\label{sec:confidence}

WAKE's mathematical concepts and implementation can be
    extended to offer confidence intervals
in addition to mean estimates.


\paragraph{Extended Definition}
\system maintains the ``uncertainty'' of all mutable attributes throughout processing in three steps: (1) it computes the uncertainty of initial mutable attributes, (2) propagates the uncertainty through \mdf operations, and 
    (3) derives CIs from the final uncertainty. 
Specifically, we extend the \mdf definition
    (\cref{sec:mdf}):
\begin{equation}
    \texttt{df\_ci  := (list((attr1, attr2, \dots, attrM)), }\Sigma \texttt{)}
\end{equation}
where 
a \emph{covariance matrix} 
$\Sigma$ captures the uncertainty
    with $\Sigma_{i,j}$ denoting 
the covariance between mutable attributes 
    attr$_i$ and attr$_j$. 

\paragraph{Initial Variance} 
When mutable attributes first appear,
\system infer variances
using the existing aggregation-specific variance estimators. 
Specifically, it measures the variance of 
cardinality power $\Var(w)$ by 
    calculating the variance of ordinary least square parameter~\cite{hayashi2000finiteols}; 
sum and avg by applying central limit theorem~\cite{polya1920CLT}; 
count-distinct using Poissonization on empirical density function~\cite{Bogachev2008variancecountdistinct};
order statistics using bootstrapping~\cite{efron1979bootstrap}; and 
extreme order statistics (min/max) by 
    fitting generalized extreme value distribution~\cite{kotz2000GEV}.

\paragraph{Variance Propagation}
\system propagate $\Sigma$ through edf operations using a standard statistics technique: ``propagation of uncertainty''~\cite{Ku2010NotesPropError}. For a differentiable mapping $v = f(u)$ and known covariance matrix $\Sigma^U$, \system linearizes $f$ using first-order Taylor expansion and computes $\Sigma^V = J \Sigma^U J^\T$ where $J$ is the Jacobian matrix ($J_{i,k} = \partial f_i / \partial U_k$). \cref{eq:propagation-of-uncertainty} expands this expression.
\begin{equation} \label{eq:propagation-of-uncertainty}
    \Sigma^V_{i,j} = \sum_k \sum_l 
    \; \Sigma^U_{k, l} 
    \; (\partial f_i / \partial U_k) \; (\partial f_j / \partial U_l)
\end{equation}

\noindent
For instance, the covariance matrix propagates through \texttt{count}
and \texttt{sum} as follows, respectively:
\begin{align}
    \Var(f_{\text{count}}) 
    &= \Var(\xh_{i, :t}) = (\xh_{i, :t} \ln(1/t))^2 \Var(w) \\
    \Var(f_{\text{sum}}) &= 1 / x_{i,t}^2 \left( \Var(y_{i,t}) \xh_{i, :t}^2 + \Var(\xh_{i, :t}) y_{i,t}^2 \right)
\end{align}
Please see our extended manuscript~\cite{deepola-tech} for other operations (e.g.,
    weighted avg, count distinct, order statistics, map/projection).
\paragraph{Variance-based Confidence Interval}
Finally, \system derives a CI of an estimate $y$ from its variance $\sigma^2 = \Var(y)$ based on Chevbyshev's inequality~\cite{chebychev1867ineq}. It outputs $[y - k \sigma, y + k \sigma]$ where $k = \sqrt{1 / (1 - \delta)}$ for confidence level $(1 - \delta)$. For example, $k \approx 4.5$ for 95\% CI.

\paragraph{Limitations} 
The above method applies to differentiable operations,
    which include the most  we are interested in.
Like the mean estimates, finite-variance uncertainty
    also assumes that the distribution of the observed data 
        represents that of the unobserved,
a fundamental premise of machine learning and
    statistical inference.
CI calculation incurs time and memory overheads to compute \cref{eq:propagation-of-uncertainty}; however,
these overheads are relatively small for TPC-H queries 
    because only  a small number of covariances are relevant.

\paragraph{Alternatives} 
The variance propagation can be substituted with 
higher moments for enhanced accuracy;
however, its time complexity (and runtime overheads) 
    increases.
Other alternatives, such as bootstrapping or 
    the propagation of parametric distributions,
either incur significant overhead or 
are applicable to limited operations.




\newcommand{\eof}{\textsf{EOF}\xspace}
\newcommand{\queryservice}{\textsf{Query Service}\xspace}
\newcommand{\executionengine}{\textsf{Execution Engine}\xspace}

\section{Implementation}
\label{sec:system}
\system's implementation in Rust can be majorly divided into two parts:
    (1) \queryservice which lets the user build a query and 
    (2) \executionengine which executes the built query in an OLA manner.

\subsection{\queryservice}
Users express a query as an execution graph composed of \textit{nodes} representing different operations, and \textit{edges} representing the data flow path between these nodes. 
A node has as many incoming edges (representing the operation's inputs) as the number of arguments appearing in an operation. 
For example, a \textsf{join} operation requires two incoming edges representing the \mdf{}s to be joined, whereas a basic \textsf{filter} operation requires one incoming edge.
To support the \mdf operations described in \cref{sec:mdf:op}, 
    \system implements different node types 
        such as \textsf{reader}, \textsf{merge-join}, \textsf{hash-join},
    \textsf{aggregator}, etc., allowing its users to express a large variety of queries. 
%
The edges are implemented using \textsf{channels} for sending a stream of messages across threads.
The user can incrementally add nodes and edges to the query graph representing nested OLA ops.

The circles in \cref{fig:implementation-query-example} represent the nodes and the green arrows represent the edges. 
The \textsf{LI}---\textsf{read_csv}---in the figure
    serves as the root,
        which passes fetched partitions
    to its subscriber (i.e., \textsf{OQ}),
        which continues as defined in the graph.
The current leaf node (i.e., \textsf{TC}) represents the query output,
        which can be consumed by downstream applications (e.g., progressive visualization).



\begin{figure*}[ht]
\centering
\pgfplotstableread{
query	polars	presto	postgres	wake-final	wake-first
1	45.378	23.9058	78.65	10.89	0.32
2	4.286	7.7356	16.85	10.32	7.8
3	31.152	38.9243	49.96	10.27	0.66
4	18.336	23.8792	48.59	5.36	0.23
5	33.223	64.8515	167.3	10.77	0.32
6	24.941	8.4505	30.48	5.81	0.15
7	40.577	73.2183	223.32	9.75	0.5
8	40.296	109.9238	140.46	11.92	1.2
9	69.608	144.8978	94.23	41.7	3.84
10	35.461	54.037	52.65	30.06	1.8
11	5.222	7.2693	23.49	2.15	0.79
12	30.663	15.524	43.76	7.84	0.23
13	32.001	20.3074	24.77	35.03	3.78
14	26.504	9.7277	32.82	8.6	0.55
15	42.331	17.5047	46.11	10.27	0.16
16	3.794	5.2871	17.8	16.65	4.74
17	31.619	44.5593	141.01	86.95	83.88
18	25.737	96.2936	164.74	13.82	0.49
19	48.689	23.1446	34.71	33.75	1.55
20	35.071	26.9678	803.71	87.31	2.09
21	55.721	172.5658	245.21	15.09	0.63
22	7.783	5.2246	3.32	13.89	1.23
}\ExactLatencyScaleFifty

\pgfplotstableread{
query	postgres_mean	postgres_std	presto_mean	presto_std	polars_mean	polars_std	polars-mem_mean	polars-mem_std	wake-final_mean	wake-final_std	wake-first_mean	wake-first_std	wake-second_mean	wake-second_std	wake-mem_mean	wake-mem_std	vertica_mean	vertica_std	actian_mean	actian_std	actian_wo_mean	actian_wo_std
1	182.0411	1.351252789	48.26	0.52	21.008	0.181585609	45.5912916	0.872627832	19.46022908	0.350938653	0.4349436	0.040412366	0.570314458	0.047230094	15.33433855	4.147169388	23.3817018	0.100751719	9.740866	0.01946244867	10.3227156	0.03180991092
2	39.5294	0.980933818	23.11	0.47	2.077	0.033015148	2.9268384	0.013304207	2.094680032	0.06584204	2.094680032	0.06584204	2.081676033	0.078137716	5.434282182	0.268734411	11.0647318	0.1193756603	0.4186795	0.01567926153	5.0163372	0.02466440127
3	109.9819	2.295597547	83.31	0.65	10.737	0.130388309	27.00161	0.031622924	15.8974662	0.133167874	0.842835865	0.11763537	0.906652674	0.117437234	5.889857818	0.437371788	19.9674706	0.8777221036	7.1816466	0.02325063146	29.9051334	0.07701924726
4	117.8691	3.218945461	52.68	1.14	8.129	0.095388096	12.333088	0.056148281	7.401458116	0.128267994	0.221764159	0.023397286	0.252892647	0.041307267	3.391436	0.197428391	23.6728597	0.336831797	1.4211187	0.01092069298	23.2930912	0.09482039816
5	3781.2457	264.2223332	138.81	2.11	10.165	0.115686358	26.0368704	0.060931397	23.06639352	0.192653404	0.686326571	0.143410154	0.858547429	0.143270887	6.783041818	1.081891748	18.8454344	0.3522859255	5.9725366	0.01202102156	28.3696244	0.07214350503
6	68.1121	1.608339755	16.6	0.57	7.026	0.0665332	22.7101208	0.065499272	8.518824098	0.124431069	0.103176345	0.007966327	0.158844563	0.018709996	1.304637091	0.059150268	16.4645385	0.6262723272	7.7021782	0.01531916458	7.729628	0.01701684764
7	1458.5954	42.61756647	164.55	5.38	0	0	0	0	12.37770382	0.182354161	0.671568348	0.24568835	0.707467157	0.2360739	4.732300727	0.243454725	20.1382418	0.5755270818	6.3922486	0.01793506703	29.2495188	0.2010660931
8	332.6001	8.792307603	247.92	1.93	19.411	0.115993295	37.806682	0.075446268	24.55403745	0.242885096	1.557477017	0.25306501	1.680356128	0.26062544	13.31928945	1.540805074	23.2121521	0.7829464039	7.1974192	0.005453933751	34.4987468	0.1125507517
9	191.8368	1.977745057	321.27	4.14	0	0	0	0	79.49075433	0.422076848	4.558736101	0.322676269	5.208640775	0.302030041	25.754244	0.879286258	89.7018787	0.9753448229	11.4442936	0.04077203503	47.1228594	0.1868583309
10	109.7914	1.203451998	109.34	1.53	14.518	0.21596296	22.4570744	1.06019971	109.7266176	0.978550864	3.670756696	0.211895509	3.85876055	0.192030372	19.03323091	0.655333897	48.253437	1.002184127	17.4398776	0.7151059733	64.4122994	0.8697143077
11	54.9514	1.21515561	13.79	0.15	2.637	0.04667857	3.7102888	0.036161912	3.151594917	0.05411685	0.114406764	0.009628911	0.191436908	0.010922832	3.376234182	0.091772251	2.675829	0.03100795746	0.4309528	0.01394210287	2.000592	0.01573850982
12	93.6888	0.868997098	31.98	0.75	9.113	0.058509259	28.1392904	0.21304661	13.58349239	0.22160512	0.209299683	0.010215007	0.336138598	0.064916939	5.8067	0.248205247	21.857195	0.09374773954	2.1332451	0.02121337014	25.6695238	0.2339499681
13	56.7403	1.493829162	53.65	0.55	23.034	0.210089928	32.9666496	0.105466756	139.1003625	1.544304392	6.99054686	0.815944822	12.54319974	1.312600243	29.84551927	0.678785276	29.8452849	0.6051261781	5.1479392	0.01729542707	77.114077	0.5503126002
14	65.551	0.426614059	22.66	0.45	8.466	0.051467358	24.2324136	0.122512766	13.37254577	0.115165674	0.956298269	0.053883118	1.03075781	0.047591573	3.171357091	0.040823785	22.3460716	0.2467499656	6.108706	0.01473433435	10.790304	0.06412858849
15	132.7422	0.508144511	36.04	0.64	13.627	0.112551223	24.0322812	0.094329276	32.97110674	0.25568995	0.226519122	0.020768969	0.349602852	0.060426228	6.424643273	0.123547628	23.8561444	0.7232244089	5.7496561	0.009949495469	5.7769134	0.02687862418
16	51.9492	0.262503249	7.94	0.14	3.225	0.076194196	3.5541992	0.046188676	28.24876391	0.483787305	1.558853896	0.121499757	1.896943985	0.13370134	10.26218473	0.546036964	4.8108355	0.1762994067	0.542677	0.01662948999	4.1748996	0.02110086639
17	300.557	8.546759204	104.59	10.17	20.953	0.099000561	34.1108424	0.082091633	10.81740671	0.084496229	10.66225953	0.084973249	10.6635859	0.084814541	2.797753091	0.165653612	19.9641511	0.3387594785	6.8464414	0.01021931512	12.4517064	0.1580358586
18	376.8785	6.599011563	200.73	3.18	45.408	0.58440473	44.7904536	0.579118024	29.73917339	0.299842211	0.958993837	0.240856448	1.078369802	0.215195856	12.69195927	0.185714183	175.9816978	1.790988868	5.4805551	0.03430677133	36.9930468	0.1086266338
19	62.303	2.464419382	46.36	0.43	14.297	0.207634829	46.6258264	0.102098638	20.36112868	0.140976837	0.852750759	0.12444021	0.94096397	0.15208301	4.477033455	0.158982902	29.1094704	0.4630159726	8.8011425	0.01239429754	20.9570276	0.09381609412
20	2785.0508	111.6741733	58.17	0.51	19.827	0.249802144	23.2947208	0.122434137	379.7484911	5.405191948	2.31424476	0.225667279	2.314042471	0.290401369	21.6566632	0.392652355	27.0499598	0.491718655	6.13621	0.01137383315	21.0568248	0.1165011077
21	1061.2006	29.39361766	352.77	9.01	88.599	0.485350961	45.2795452	0.283629268	30.61035963	0.215628451	0.705807656	0.089498496	1.147408431	0.187662778	13.7142244	0.472227103	84.7812601	1.324885656	5.4919648	0.06102268128	25.5845268	0.109118782
22	6.3497	0.15225421	9.28	0.24	14.492	0.201539	6.1808916	0.026219767	46.27614936	0.422253256	1.542939738	0.06895119	2.418432479	0.077470547	8.7027496	0.165874905	13.0869127	0.09954275913	0.6150123	0.01868332382	7.9358778	0.3770330829
}\ExactLatencyScaleHundred

\begin{tikzpicture}
\begin{axis}[
    height=37mm,
    width=\linewidth,
    ymode=log,
    ymin=0.1,
    ymax=10000,
    log origin=infty,
    bar width=2pt,
    ybar=0.2pt,
    enlarge x limits=0.025,
    xlabel=TPC-H Query Number (of original queries),
    xlabel near ticks,
    xlabel shift=-1mm,
    xtick = {1,2,3,4,5,6,7,8,9,10,11,12,13,14,15,16,17,18,19,20,21,22},
    xticklabels = {q1,q2,q3,q4,q5,q6,q7,q8,q9,q10,q11,q12,q13,q14,q15,q16,q17,q18,q19,q20,q21,q22},
    ytick = {0.1, 1, 10, 100, 1000, 10000},
    yticklabels = {100ms, 1s, 10s, 100s, 1000s, 10000s},
    xticklabel style = {yshift=1mm,font=\sf\footnotesize},
    xtick pos=left,
    ylabel=\textsf{Query Latency},
    ylabel near ticks,
    ylabel style={align=center},
    label style={font=\footnotesize\sf},
    ylabel shift=-2mm,
    legend style={
        at={(0.5, 0.95)},anchor=south,column sep=2pt,
        draw=black,fill=white,line width=.5pt,
        font=\scriptsize,
        /tikz/every even column/.append style={column sep=5pt}
    },
    legend columns=-1,
    colormap name=bright,
    every axis/.append style={font=\footnotesize},
    ymajorgrids,
    ylabel near ticks,
    legend image code/.code={%
        \draw[#1, draw=none] (0cm,-0.1cm) rectangle (0.6cm,0.1cm);
    },
]

\addplot plot [
    fill=BlueColor,draw=none,error bars, y dir=both,y explicit]
table[x=query,y=postgres_mean,y error=postgres_std] {\ExactLatencyScaleHundred};
\addlegendentry{PostgreSQL}

\addplot plot [
    fill=YellowColor,draw=none,error bars, y dir=both,y explicit]
table[x=query,y=presto_mean, y error=presto_std] {\ExactLatencyScaleHundred};
\addlegendentry{Presto}

\addplot plot [
    fill=GreenColor,draw=none,error bars,y dir=both,y explicit
] 
    table[x=query,y=vertica_mean, y error=vertica_std] {\ExactLatencyScaleHundred};
\addlegendentry{Vertica}

\addplot plot [
    fill=RedColor,draw=none,error bars, y dir=both,y explicit
    ] 
    table[x=query,y=polars_mean, y error=polars_std] {\ExactLatencyScaleHundred};
\addlegendentry{Polars}

\addplot plot [
    fill=PurpleColor,draw=none,error bars, y dir=both,y explicit
    ] 
    table[x=query,y=actian_wo_mean, y error=actian_wo_std] {\ExactLatencyScaleHundred};
\addlegendentry{Actian Vector}

\addplot plot [
    fill=BrownColor,draw=none,error bars,y dir=both,y explicit
] 
    table[x=query,y=wake-final_mean, y error=wake-final_std] {\ExactLatencyScaleHundred};
\addlegendentry{\system{}\textbf{-final} (ours)}

\addplot plot [
    fill=PinkColor,draw=none,error bars,y dir=both,y explicit
] 
    table[x=query,y=wake-second_mean, y error=wake-second_std] {\ExactLatencyScaleHundred};
\addlegendentry{\system{}\textbf{-first} (ours)}

\end{axis}
\end{tikzpicture}

\vspace{-2mm}
\caption{Comparison of different baselines on TPC-H 100 GB dataset. The results are averaged across 10 runs.
}
\label{fig:exact-total-latency}
\vspace{-2mm}
\end{figure*}
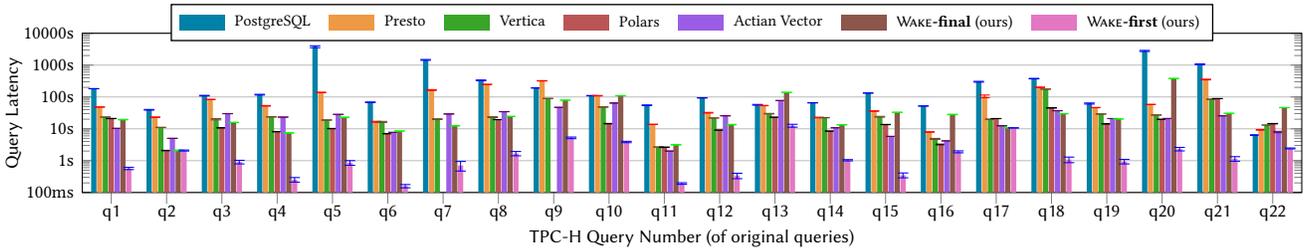
\begin{figure*}[ht]
    \centering
    \input{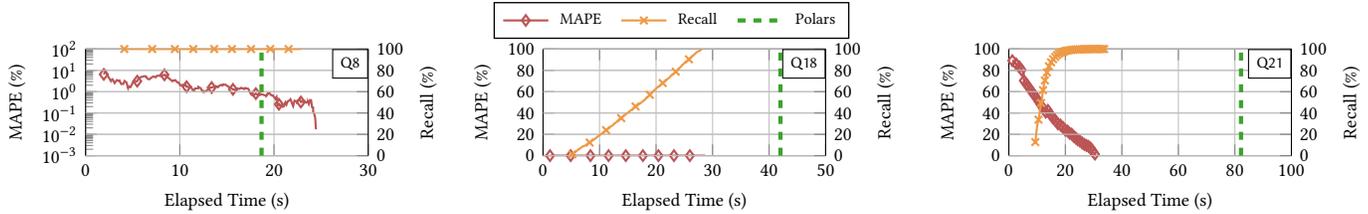}
    \def\subplotwidthThree{0.3\linewidth}
    \def\height{30mm}
    \def\hspaceGap{0.04\linewidth}
    \pgfplotsset{
        accfigLeft/.style={
            height=\height,
            width=\linewidth,
            ymode=linear,
            ymin=0.0,
            ymax=100.0,
            xmin=0,
            ytick={0, 20, 40, 60, 80, 100},
            axis y line*=left,
            xlabel=Elapsed Time (s),
            xlabel near ticks,
            ylabel=MAPE (\%),
            ylabel near ticks,
            ylabel style={align=center},
            legend columns=2,
            colormap name=bright,
            every axis/.append style={font=\footnotesize},
            xmajorgrids,
            ymajorgrids,
            ylabel near ticks,
        },
        accfigLeft/.belongs to family=/pgfplots/scale,
    }
    \pgfplotsset{
        accfigRight/.style={
            height=\height,
            width=\linewidth,
            axis x line=none,
            axis y line*=right,
            enlarge x limits=0.1,
            ymin=0.0,
            ymax=100.0,
            xlabel=Elapsed Time (s),
            xlabel near ticks,
            ylabel={Recall (\%)},
            ylabel near ticks,
            ylabel style={align=center},
            legend style={
                at={(1.45,0.5)},anchor=west,column sep=2pt,
                draw=black,fill=white,line width=.5pt,
                font=\scriptsize,
                /tikz/every even column/.append style={column sep=5pt}
            },
            legend columns=1,
            colormap name=bright,
            every axis/.append style={font=\footnotesize},
            every non boxed x axis/.append style={x axis line style=-},
            xmajorgrids,
            ymajorgrids,
            ylabel near ticks,
        },
        accfigRight/.belongs to family=/pgfplots/scale,
    }
    \begin{subfigure}[b]{\linewidth}
    \begin{tikzpicture}
        \begin{axis}[
                ticks=none,
                width=\linewidth,
                hide axis,
                xmin=10,
                xmax=50,
                ymin=0,
                ymax=0.4,
                legend style={
                    at={(0.55,0)},
                    anchor=center,
                    column sep=2pt,
                    draw=black,fill=white,line width=.5pt,
                    font=\scriptsize,
                    /tikz/every even column/.append style={column sep=5pt}
                },
                legend columns=-1,
            ]
            
            \node[align=center, opacity=0] {
            \addlegendimage{wakecolor,mark=diamond,thick}
            \addlegendentry{MAPE};
            \addlegendimage{wakerecallcolor,mark=x,thick}
            \addlegendentry{Recall};
            \addlegendimage{polarcolor,ultra thick,dashed}
            \addlegendentry{Polars};
            };
        \end{axis}
    \end{tikzpicture}
    \end{subfigure}
    \vspace{-5mm}
    
    \begin{subfigure}[b]{\subplotwidthThree}
        \centering
        \begin{tikzpicture}
            \begin{axis}[
                accfigLeft,
                ymode=log,
                ymin=0.001,
                ymax=100,
                ytick={0.001, 0.01, 0.1, 1, 10, 100},
                xmin=0.0,
                xmax=30.0,
            ]
                \addplot[mark=diamond, color={wakecolor}, thick, mark repeat=20]
                table[x=time_avg_s,y=mape_p] {\powerHundredqVIII}; \label{fig:II_accuracy_rp_curve_A_mape}
            \end{axis}
            
            \begin{axis}[
                accfigRight,
                xmin=0.0,
                xmax=30.0,
            ]
            
            
                \addplot[mark=x, color={wakerecallcolor}, thick,  mark repeat=20]
                table[x=time_avg_s,y=recall_p] {\powerHundredqVIII};
                
                \node[draw=black, fill=white, anchor=north east, align=right, font=\scriptsize\sf] 
                    at (rel axis cs:1.00,1.00) {Q8};
                
                \addplot[ultra thick, dashed, samples=50, smooth,domain=0:6,color={polarcolor}] coordinates {(19.411,0)(19.411,100)};
            \end{axis}
        \end{tikzpicture}
    \end{subfigure}
    \hspace{\hspaceGap}
    \begin{subfigure}[b]{\subplotwidthThree}
        \centering
        \begin{tikzpicture}
            \begin{axis}[
                accfigLeft,
                ymode=linear,
                ymin=0.0,
                ymax=100.0,
                xmin=0.0,
                xmax=50.0,
                xtick={0, 10, 20, 30, 40, 50},
            ]
                \addplot[mark=diamond, color={wakecolor}, thick, mark repeat=20]
                table[x=time_avg_s,y=mape_p] {\powerHundredqXVIII}; \label{fig:II_accuracy_rp_curve_B_mape}
            \end{axis}
            
            \begin{axis}[
                accfigRight,
                xmin=0.0,
                xmax=50.0,
            ]
            
            
                \addplot[mark=x, color={wakerecallcolor}, thick, mark repeat=20]
                table[x=time_avg_s,y=recall_p] {\powerHundredqXVIII};
                
                \node[draw=black, fill=white, anchor=north east, align=right, font=\scriptsize\sf] 
                    at (rel axis cs:1.00,1.00) {Q18};
                
                \addplot[ultra thick, dashed, samples=50, smooth,domain=0:6,color={polarcolor}] coordinates {(45.408,0)(45.408,100)};
            \end{axis}
        \end{tikzpicture}
    \end{subfigure}
    \hspace{\hspaceGap}
    \begin{subfigure}[b]{\subplotwidthThree}
        \centering
        \begin{tikzpicture}
            \begin{axis}[
                accfigLeft,
                xmin=0.0,
                xmax=100.0,
            ]
                \addplot[mark=diamond, color={wakecolor}, thick, mark repeat=5]
                table[x=time_avg_s,y=mape_p] {\powerHundredqXXI}; \label{fig:II_accuracy_rp_curve_C_mape}
            \end{axis}
            \begin{axis}[
                accfigRight,
                xmin=0.0,
                xmax=100.0,
            ]
            
            
                \addplot[mark=x, color={wakerecallcolor}, thick, mark repeat=5]
                table[x=time_avg_s,y=recall_p] {\powerHundredqXXI};
                
                \node[draw=black, fill=white, anchor=north east, align=right, font=\scriptsize\sf] 
                    at (rel axis cs:1.00,1.00) {Q21};
                
                \addplot[ultra thick, dashed, samples=50, smooth,domain=0:6,color={polarcolor}] coordinates {(88.599,0)(88.599,100)};
            \end{axis}
        \end{tikzpicture}
    \end{subfigure}
    \vspace{-2mm}
    \caption{\system's approximation error measured in mean absolute percentage error (MAPE) and recall over time. Vertical lines represent completion time of exact methods. From left to right (Q8, Q18, Q21), \postgres completes in (332s, 376s, 1061s), while \presto completes in (247s, 200s, 352s) respectively.}
    \label{fig:II_accuracy_rp_curve}
    \vspace{-2mm}
\end{figure*}

\subsection{Execution Engine}
The \executionengine takes a \queryservice and evaluates the query on a specified dataset generating a sequence of \mdf outputs.
On specifying the input for each \textsf{read_csv} node, the execution engine starts the query execution.
Each node operates in a separate thread,
    reading messages from its input channels.
A received message consists of: (1) a shared pointer to a data frame and (2) metadata containing information on the progress of the query execution (\cref{sec:processing}). The node processes these messages, updates its intrinsic states using the metadata, and writes its extrinsic states along with the metadata as a message to the output channels. In case there are no messages on a node's input channel, the node blocks on the channel read.
A special message type---\eof---is used to indicate the end of inputs on a given channel. 
Once an execution node receives \eof messages on its input channels, the node sends an \eof message on its output channels and terminates its execution. 





\subsection{Discussion}
\label{sec:discussion}

\paragraph{Query Optimization.}
As its first step, \system introduces Deep OLA primitives without 
    a declarative language.
We plan to adopt existing optimization techniques
    such as predicate pushdown, join order optimization, etc.,
but will also investigate unique opportunities.
For example, join algorithms (e.g., sort-merge or hash)
    affects not only the performance but the way that intermediate results are delivered.
If a subsequent group-by uses the same key as a join,
    we may opt for merge join for more interactive results
        even if a hash join can produce final results more quickly.
        


\paragraph{OLA-Specific Optimizations.}
\system's design includes OLA-specific optimizations such as (1) pipelined implementation of a query's operations, 
(2) sort-merge join when both the tables are partitioned on a common clustering key (e.g., \textsf{lineitem} and \textsf{orders}), 
(3) re-using the hash-table of right (build) tables for repeated hash-join, 
(4) shared pointers of data to reduce cloning costs, etc. 
These optimization help \system produce exact answers
    as quickly as other systems designed for exact query processing.

\paragraph{Intra-Query Parallelism.}
\system benefits from both data pipelining and multi-threading,
    which are widely employed in data systems to reduce task completion time~\cite{meehan2015s,toshniwal2014storm}.
Our extended report~\cite{deepola-tech} studies the benefit of data pipelining.
In the future, we will investigate Deep OLA-specific distributed processing.


\section{Evaluation}
\label{sec:exp}

This section evaluates \system
    against (conventional) exact data systems 
        as well as OLA systems. 
Our experiments show the following:
\begin{itemize}
\item \system's first estimate is \mediantimebest faster than the exact systems
    while being \medianoverheadbest slower in producing exact results. (\cref{sec:exp:latency})
\item \system's first estimates have a median error of \medianfirsterror. \system provides results with under 1\% error \medianmultipleonepercent faster on average (upto \maxmultipleonepercent) 
    than the best exact data system. (\cref{sec:exp:error})
\item \system produces the results with less than 1\% error, \medianmultipleonepercentola times faster
    than state-of-the-art OLA systems. (\cref{sec:exp:ola})
\item \system's CIs comply with the chosen level of confidence but can be conservative towards the end of processing. (\cref{sec:exp:ci})
\item \system executes deep queries at expected time complexity. (\cref{sec:ablation-synthetic-deep-queries})
\item \system's performance can be further improved by optimally choosing the partition sizes. (\cref{sec:ablation-partition})
\end{itemize}

\subsection{Experimental Setup}
All the experiments are performed on a Standard D16ads v5 (Azure) machine with 16 vCPU(s) and 64 GB of memory. The TPC-H Benchmark is used to compare and evaluate the different systems.


\paragraph{Baselines} We employ 2 state-of-the-art OLA and 4 exact systems.
\begin{enumerate}

\item \progressivedb~\cite{berg2019progressivedb}: 
A middleware-based OLA system for non-nested, join-free queries.
We use the authors' implementation~\cite{progressivedb-implementation} (currently limited to a single table) and evaluate on TPC-H queries (Q1, Q6) 
    for benchmark (\cref{fig:II_accuracy_curve_progressivedb}).
        

\item \wanderjoin~\cite{li2016wander}: 
A random walk-based OLA system for non-nested, multi-join queries.
We use the authors' implementation~\cite{wanderjoin-implementation} 
    with their modified queries (Q3, Q7, Q10) (\cref{fig:II_accuracy_curve_wanderjoin}).

\item \presto~\cite{sethi2019presto}: 
    \presto is a data warehouse
        designed for diverse data sources.
We use its Hive connector on HDFS.

\item \postgres: A popular RDBMS. 
We create appropriate FKs and indexes on the attributes according to the TPC-H schema.

\item \polars~\cite{polars}: \polars is a data frame library in Rust
optimized with SIMD, Arrow~\cite{arrow}, 
    lock-free parallel hashing, etc. 

\item \vertica~\cite{lamb2012vertica}: \vertica is a columnar storage-based analytical database. 
We use Vertica Data Warehouse on Azure.

\item \actianvector~\cite{zukowski2012vectorwise}: \actianvector is a high-performance vectorized analytical database.

\end{enumerate}

\paragraph{File Format} 
For \presto, \polars, and \system, we use Parquet~\cite{parquet}. 
\wanderjoin and \progressivedb are implemented on top of \postgres.

\paragraph{Dataset} 
We use a scale-100 (100 GB) TPC-H dataset~\cite{tpch}.
%
%
For \system,
    the dataset is partitioned into 512 MB chunks,
each of which is then converted to Apache Parquet format.

\paragraph{Queries} 
We use the 22 TPC-H queries to evaluate the different systems and compare their performance. We employ TPC-H queries for two reasons. First, consistency with the existing work for accurate comparison. 
Second, many TPC-H queries can be considered \textit{deep} since existing OLA methods cannot handle them 
    due to nested \texttt{select} statements --- except for Q1 and Q6.
Relatedly, our comparison against existing OLA uses modified Q3, Q7, and Q10 (\cref{sec:exp:ola});
however, for studying our system in \cref{sec:exp:latency} and \cref{sec:exp:error},
    we use the original queries without simplifications.
Finally, we evaluate \system additionally with systematically generated deeper queries (\cref{sec:ablation-synthetic-deep-queries}).

\paragraph{Metrics}
The following metrics are computed:
\begin{itemize}
\item Final-Result Latency: The time taken to process the complete dataset and produce a correct final result.

\item First-Estimate Latency: The time for the first estimate.

\item Peak-Memory Usage: For in-memory \polars and \system, we compute the peak-memory usage (i.e. maximum resident-set size).

\item MAPE: Mean Absolute Percentage error is used to calculate the approximation error. 

\item Recall: For group-by queries, recall measures the fraction of final-result groups that were correctly produced.
\end{itemize}




\input{figures/experiments/II-2_ola_approx_error_trend}

\subsection{Interactive Querying \& Low Overhead}
\label{sec:exp:latency}

We compare \system's latency against different exact query processing systems in terms of the time taken to produce first estimates and final exact results. 
\cref{fig:exact-total-latency} shows the time taken by \system 
    to obtain the first and the final result. 
The number of intermediate results varies depending on the number of partitions of the tables the query uses. 
\system produces first estimates \mediantimeactian faster than \actianvector's exact answers, \mediantimepolars faster than \polars's exact answers, \mediantimevertica faster than \vertica's exact answers, \mediantimepresto faster than \presto's exact answers 
    and \mediantimepostgres faster than \postgres's exact answers (all median).
In terms of slowdown, 
    measured as the ratio of \system's final-result latency and 
other baseline's final-result latency, 
    the median slowdown against 
    \actianvector is \medianoverheadactian,
    against \polars is \medianoverheadpolars, 
    against \vertica is \medianoverheadvertica,
    against \presto is \medianoverheadpresto and against \postgres is \medianoverheadpostgres,
        meaning \system produces exact answers 
    even faster than \vertica, \presto and \postgres.
\system---despite being an OLA system---produces exact results
    more quickly 
than \postgres in 20 queries, \presto in 17 queries, and \polars in 6 queries.

Moreover,
    \system has low peak memory utilization.
\polars runs out of memory (on the experiment machine) 
    for queries Q7 and Q9, not able to generate results, 
whereas \system successfully produces exact query results. 
On average, \system's peak memory usage is \avgmemory less than \polars (up to \maxmemory less for some queries), providing the ability to handle larger datasets.

Specifically looking at some of the queries, 
    Q9, Q10, and Q13 require 
building hash tables 
    for smaller right tables 
    before being able to produce first-result, 
thus have smaller improvement. 
Q2 and Q17 require computing sub-queries' aggregate and 
    thus have negligible gains (but almost zero overhead).
In terms of total query latency, 
    computational overhead is most prominent in Q10 and Q13 (due to group-by on high-cardinality $c\_custkey$) 
and Q20 (due to repeated filter on $partsupp$).
For the first-estimate latencies,
    \system has a median error of \medianfirsterror. 
\system provides estimated results with less than 1\% MAPE, 
    \medianmultipleonepercent faster than the final result time of the best baseline.
WAKE's fast query performance benefits from
    our manual optimization such as predicate pushdown
and careful choices of join methods (\cref{fig:implementation-query-example}); expectedly, we have observed
    poorer latencies with inferior query plans such as
filtering after joins. This motivates our future work as discussed in \cref{sec:discussion}.

\subsection{\system's Approximation Error Analysis}
\label{sec:exp:error}

In this experiment, we analyze approximate errors
    of \system's OLA outputs (as it processes more data over time)
    in terms of MAPE and recall error.
\cref{fig:II_accuracy_rp_curve} shows time-error curves for a few representative queries in three  different categories, as follows.


The first category includes queries on non-clustering group-by keys with low cardinality. 
Overall, their MAPE curves decrease over time 
    as \system observes more data 
    while recalls reach 100\% early on. 
Many queries in TPC-H fall into this category: Q1, Q4--Q9, Q12, Q14, Q17, Q19, and Q22. 
In particular, Q8 (\cref{fig:II_accuracy_rp_curve}-left) involves 
    a weighted average group-by aggregation over multiple joined tables. 
\system is able to answer the first estimate at 1.9s with a 6.5\% error. 
    When \polars completes (at 19s), \system has 0.87\% error.


The queries in the second category
    involves
    clustering group-by keys; 
therefore, their aggregation values are exact (MAPE at 0\%) while their recalls increase as \system retrieve keys in different partitions. 
Q3, Q18 (\cref{fig:II_accuracy_rp_curve}-right), and Q20 are examples:
    their recalls increase linearly 
as more keys are retrieved/observed.


Third, 
the third category is a combination of the first two;
    that is, their errors 
can be understood with 
    MAPE, recall, and/or precision. 
For instance,
    Q10, Q16, and Q21 (\cref{fig:II_accuracy_rp_curve}-right) 
        have quickly rising recall curves 
    while their MAPE curves drop only linearly 
because their group-by keys are diverse, leading to a lower number of samples per group 
    and so lower prediction power.
While figures are omitted,
    Q11 has a perfect MAPE score but its recall/precision curves increase 
        quickly toward the end. 
Q2 and Q15 have on-off recall and precision due to their uses of arguments of the minima and maxima.
%
Finally,
    Q13 computes count over a high-cardinality non-clustering key (\textsf{c_custkey}), 
    followed by an outer group-by over the inner count. 
Because the inner count changes over different partitions, the growth within outer groups can be non-monotonic, 
    violating \system's cardinality estimator and resulting in a relatively large MAPE. 
We will address this limitation in the future.

\begin{figure}[t]
    \centering
    \input{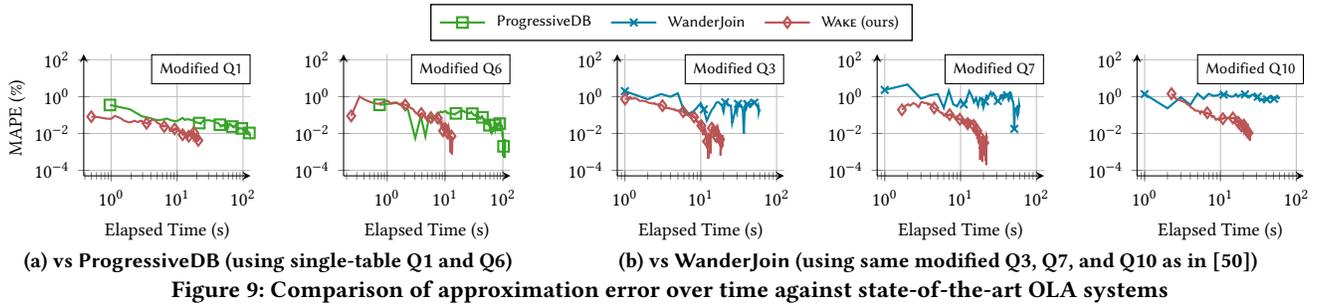}
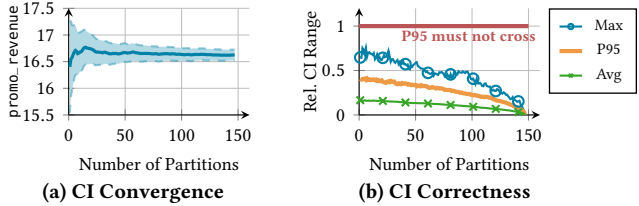
    \def\subplotwidthA{0.47\linewidth}
    \def\subplotwidthB{0.47\linewidth}
    \def\height{30mm}
    \def\hspaceGap{0.00\linewidth}
    \pgfplotsset{
        accfigLeft/.style={
            height=\height,
            width=\linewidth,
            ymode=linear,
            xmin=0,
            xmax=160,
            axis lines=left,
            xlabel=Number of Partitions,
            xlabel near ticks,
            ylabel near ticks,
            ylabel shift=-2mm,
            ylabel style={align=center},
            colormap name=bright,
            every axis/.append style={font=\footnotesize},
            xmajorgrids,
            ymajorgrids,
            ylabel near ticks,
            legend columns=1,
            legend style={
                at={(1.05,1.0)}, anchor=north west,
                column sep=2pt,
                draw=black,fill=white,line width=.5pt,
                legend image post style={scale=0.5},
                font=\scriptsize,
                /tikz/every even column/.append style={column sep=5pt}
            },
        },
        accfigLeft/.belongs to family=/pgfplots/scale,
    }
            
    \hspace{-8mm}
    \begin{subfigure}[b]{\subplotwidthA}
        \centering
        \begin{tikzpicture}
            \begin{axis}[
                accfigLeft,
                ylabel=\texttt{promo_revenue},
                ymin=15.5,
                ymax=17.5,
            ]
                \addplot[name path=cipos, mark=none, color=BlueColor, thick, dashed,  mark repeat=20, opacity=0.4]
                table[x=idx,y=x_pos] {\wakeHundredqXIVc};

                \addplot[mark=none, color=BlueColor, very thick,  mark repeat=20]
                table[x=idx,y=x] {\wakeHundredqXIVc};
            
                \addplot[name path=cineg, mark=none, color=BlueColor, thick, dashed,  mark repeat=20, opacity=0.4]
                table[x=idx,y=x_neg] {\wakeHundredqXIVc};
                
                \addplot [BlueColor!30!white] fill between [
                    of = cipos and cineg,
                    soft clip={domain=0:147},
                ];
            \end{axis}
        \end{tikzpicture}
        \vspace{-2mm}
        \caption{CI Convergence}
        \label{fig:II_ci_convergence}
    \end{subfigure}
    \hspace{\hspaceGap}
    \begin{subfigure}[b]{\subplotwidthB}
        \centering
        \begin{tikzpicture}
            \begin{axis}[
                accfigLeft,
                ylabel=Rel. CI Range,
                ymin=0.0,
                ymax=1.2,
            ]
                \addplot[mark=o, color=BlueColor, thick,  mark repeat=20]
                table[x=idx,y=rel_max] {\wakeHundredqXIVc};
                \addlegendentry{Max};

                \addplot[mark=none, color=YellowColor, ultra thick,  mark repeat=20]
                table[x=idx,y=rel_p95] {\wakeHundredqXIVc};
                \addlegendentry{P95};
            
                \addplot[mark=x, color=GreenColor, thick,  mark repeat=20]
                table[x=idx,y=rel_avg] {\wakeHundredqXIVc};
                \addlegendentry{Avg};

                \draw[RedColor,ultra thick] (axis cs: 0,1) -- (axis cs: 150,1);
                \node[font=\scriptsize\bf, anchor=north west, text=RedColor] 
                    at (axis cs: 30,1.05)
                {P95 must not cross};
                
            \end{axis}
        \end{tikzpicture}
        
        \vspace{-2mm}
        \caption{CI Correctness}
        \label{fig:II_ci_correctness}
    \end{subfigure}
    \vspace{-4mm}
    \caption{\system's 95\% confidence interval on Q14 ($k \approx 4.5$).}
    \label{fig:II_ci}
    \vspace{-3mm}
\end{figure}

\subsection{Faster \& More Accurate than Existing OLA}
\label{sec:exp:ola}

This section compares \system with other OLA systems, \progressivedb and \wanderjoin,
on all of their supported sets of TPC-H queries: Q1, Q6 for \progressivedb, and Q3, Q7, Q10 for \wanderjoin.
\cref{fig:II_accuracy_curve_progressivedb} shows the results from \progressivedb on  Q1 and  Q6. Although the initial estimates of \progressivedb and \system are close, \system converges $2.5\times$ faster than \progressivedb to a less than 1\% relative error.
\cref{fig:II_accuracy_curve_wanderjoin} shows the comparison against \wanderjoin. Although the errors of the first estimate are comparable, the convergence of \system to a less than 1\% relative error is $1.51\times$ faster than \wanderjoin.
Moreover, \system soon converges to exact answers whereas \wanderjoin stays around 1\% relative
    errors, which are expected because its random walk-based
sampling mechanism is not designed in such a way.
We believe the approach taken by \system---converging to exact answers---is 
    more desirable for end users.


\subsection{Confidence interval correctness}
\label{sec:exp:ci}
    
We empirically verify \system's CI derivation (\cref{sec:confidence}) 
    by executing Q14 with shuffled input partitions to simulate the inputs in unexpected orders. 
Q14 expresses a weighted average over a join of two tables with filters. 
Our computed confidence intervals converged toward the mean estimates as shown in \cref{fig:II_ci_convergence}.
Moreover, we investigate the quality of those confidence intervals (\cref{fig:II_ci_correctness}),
with \emph{relative CI ranges}:
the fraction of the actual absolute error over the size of the CI $|\yh - y| / (k \sigma)$. 
Initially, our P95 relative CI range is around $0.4$ as expected because our Chevbyshev-based CI assigns $k \approx 4.5$.
Later, relative CI ranges decrease over partitions; while being overly conservative, they safely bound the true answers.

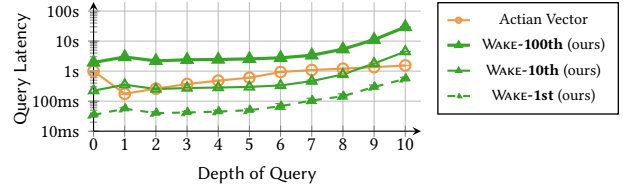
\begin{figure}[t]
\centering

\pgfplotstableread{
depth	wake-first_mean	wake-first_std	wake-mid_mean	wake-mid_std	wake-final_mean	wake-final_std	vertica_mean	vertica_std	actian_mean	actian_std
0	0.035	0.003	0.224	0.008	1.946	0.030	1	0	1	0
1	0.056	0.004	0.353	0.012	2.958	0.046	1	0	0.1761354	0.01076236898
2	0.040	0.004	0.253	0.015	2.199	0.059	1	0	0.2593864	0.01342471914
3	0.042	0.004	0.274	0.017	2.390	0.084	1	0	0.372106	0.01496805442
4	0.045	0.005	0.284	0.013	2.442	0.041	1	0	0.4825209	0.01628630432
5	0.050	0.005	0.299	0.011	2.556	0.041	1	0	0.6035954	0.01824367532
6	0.068	0.015	0.334	0.015	2.772	0.033	1	0	0.933573	0.02899947984
7	0.105	0.021	0.465	0.040	3.369	0.077	1	0	1.0889401	0.02165601342
8	0.146	0.030	0.778	0.057	5.388	0.096	1	0	1.2067621	0.03995047821
9	0.297	0.047	1.802	0.071	11.113	0.122	1	0	1.3567682	0.02091371749
10	0.555	0.068	4.478	0.070	29.574	0.167	1	0	1.5459864	0.02708663232
}\ExactLatencyScaleHundred

\begin{tikzpicture}
\begin{axis}[
    height=33mm,
    width=0.7\linewidth,
    ymode=log,
    ymin=0.01,
    ymax=200,
    xmin=0.0,
    xmax=10.5,
    log origin=infty,
    axis lines=left,
    xlabel=Depth of Query,
    xlabel near ticks,
    xtick = {0,1,2,...,10},
    ytick = {0.01,0.1, 1, 10, 100},
    yticklabels = {10ms, 100ms, 1s, 10s, 100s},
    xticklabel style = {yshift=1mm,font=\sf\footnotesize},
    ylabel=Query Latency,
    ylabel style={align=center},
    label style={font=\footnotesize\sf},
    ylabel shift=-2mm,
    legend style={
        at={(1.05, 1.0)},anchor=north west,column sep=2pt,
        draw=black,fill=white,line width=.5pt,
        legend image post style={scale=0.5},
        font=\scriptsize,
        /tikz/every even column/.append style={column sep=5pt}
    },
    legend columns=1,
    every axis/.append style={font=\footnotesize},
    xmajorgrids,
    ymajorgrids
]


\addplot[
    color=YellowColor,mark=o,thick,error bars, y dir=both,y explicit,error bar style={YellowColor}
    ] 
    table[x=depth,y=actian_mean, y error=actian_std] {\ExactLatencyScaleHundred};
\addlegendentry{Actian Vector}

\addplot[
    color=GreenColor,mark=triangle*,very thick,error bars,y dir=both,y explicit,error bar style={GreenColor}
] 
    table[x=depth,y=wake-final_mean, y error=wake-final_std] {\ExactLatencyScaleHundred};
\addlegendentry{\system{}\textbf{-100th} (ours)}

\addplot[
    color=GreenColor,mark=triangle,thick,error bars,y dir=both,y explicit,error bar style={GreenColor}
] 
    table[x=depth,y=wake-mid_mean, y error=wake-mid_std] {\ExactLatencyScaleHundred};
\addlegendentry{\system{}\textbf{-10th} (ours)}

\addplot[
    color=GreenColor,mark=triangle*,thick,dashed,error bars,y dir=both,y explicit,error bar style={GreenColor}
] 
    table[x=depth,y=wake-first_mean, y error=wake-first_std] {\ExactLatencyScaleHundred};
\addlegendentry{\system{}\textbf{-1st} (ours)}

\end{axis}
\end{tikzpicture}

\vspace{-4mm}
\caption{Impact of query depth on its performance}
\label{fig:ablation-query-depth}
\vspace{-3mm}
\end{figure}

\subsection{Performance on Synthetic Deep Queries}
\label{sec:ablation-synthetic-deep-queries}

We study how \system's internal processing scales with query complexity. We synthetically generate a 100-partition dataset with 100M rows of $11$ integer columns. Ten of which are group-by columns having $4^{10}$ unique combinations. The synthetic query alternates between summation and maximum aggregations with a \emph{query depth} $d = 0, \dots, 10$. For example, if $d = 2$, the query is \texttt{df.max(x, by=(ci,cii)).sum(max\_x, by=ci).sum(sum\_max\_x)}. Figure~\ref{fig:ablation-query-depth} shows \system's latency to 1st, 10th, and 100th (final) results on a logarithmic y-axis. Across all depths, \system takes similar times to process each partition and outputs the results at a regular pace. As the depth increases, \system execution times scale with the primary group cardinality (i.e. $O(4^d)$ groups for depth $d$) because it needs to merge the new aggregate into the existing aggregate. In general, for $n$ rows, $B$ rows per partition, and depth $d$, \system's time complexity is $O(4^d n / B + n)$ whereas non-OLA engine's complexity (e.g., Actian Vector) is expected to be $O(n)$.

\subsection{Data Partition Size Matters}
\label{sec:ablation-partition}

To understand the impact of  partition sizes on overall query latencies, 
    we evaluate \system on scale-100 (100 GB) data with 
        different partition sizes (128 MB, 256 MB, ..., 2048 MB).
As individual partition sizes increase, 
    the time taken to generate the first result increases
whereas
    the final-result latency tends to decrease 
        because the overhead of merging multiple partitions is lower.
\cref{fig:ablation-partition} shows the latencies of multiple TPC-H queries 
    as a multiple of the best performance observed for that query across different partitions.

For queries with small merge operation overhead, 
    the partition size expectedly does not affect the final-result latency. 
Some example queries are Q1, Q4, Q6, Q7, Q12, Q19 (group-by-agg has few groups), Q18, and Q21 (streamed on \textsf{o_orderkey}).

For queries with higher merge overhead (e.g., \textsf{group-by} with a large number of groups), 
    the partition size makes a significant performance difference. 
A larger size reduces the final-result latency as the number of partitions decreases and thus the overhead of computation drops. 
Some examples are Q13, Q15, Q16 (group-by-agg has large number of groups), and Q22 (pruning of \textsf{c_custkey}).


For less OLA-friendly queries (e.g., Q17), 
    a larger partition size helps in reducing both first-result and final-result latencies. 
Hence, a suitable partition size depends on the query as well as the goal---be it either to minimize first-result latency or final-result latency. 
This work, however, does not investigate such an optimizer.





\begin{figure}[t]
\centering
\pgfplotstableread{
partition	1	2	3	4	5	6	7	8	9	10	11	12	13	14	15	16	17	18	19	20	21	22
128	0.88406245	0.03415245	0.44442529	0.85583515	0.30072514	0.79196473	0.73359314	0.55183366	0.1816282	nan	0.19050782	0.81234692	0.11404295	0.44736009	0.20267103	0.11991214	0.10419508	0.81378401	0.68028819	nan	0.92631277	0.13032933
256	0.85544984	0.10708358	0.71849828	0.92311962	0.50421476	0.88238458	1	0.70682585	0.34863901	nan	0.342717	0.91586051	0.18573718	0.76220289	0.31449685	0.21553854	0.22117719	0.97691331	0.86074864	nan	0.98587764	0.24218897
512	1	0.2684672	0.91584614	0.96782314	0.75266507	0.8798406	0.85682622	0.82737177	0.55579759	nan	0.60153541	0.92950662	0.35215048	0.72778518	0.53269267	0.3950655	0.39874027	0.98990164	0.92384373	nan	1	0.42287764
1024	0.88003113	0.58808841	0.9767599	0.90890417	0.90904933	0.97877959	0.91055975	0.90763702	0.77190981	nan	0.81370939	1	0.59966538	1	0.8575894	0.64162438	0.6586474	1	0.95775679	nan	0.98211765	0.64179727
2048	0.88065137	1	1	1	1	1	0.89915323	1	1	nan	1	0.95907994	1	0.78046855	1	1	1	0.96337776	1	nan	0.94031775	1
}\ExactLatencyScaleHundredPartitionAblation

\begin{tikzpicture}
\begin{axis}[
    height=32mm,
    width=0.35\textwidth,
    ymin=0,
    ymax=1,
    xmode=log,
    xmin=128,
    xmax=2048,
    enlarge x limits=0.025,
    xlabel=Partition Size (in MB),
    xlabel near ticks,
    xtick = {128,256,512,1024,2048},
    xticklabels = {128, 256, 512, 1024, 2048},
    ytick = {0, 0.1, 0.2, 0.5, 1},
    yticklabels = {{}, 10x, 5x, 2x, 1x},
    xticklabel style = {font=\sf\footnotesize},
    ylabel=\textsf{Slowdown},
    ylabel near ticks,
    ylabel style={align=center},
    label style={font=\footnotesize\sf},
    ylabel shift=0mm,
    legend style={
        at={(1.02, 1)},anchor=north west,column sep=2pt,
        draw=black,fill=white,line width=.5pt,
        font=\scriptsize,
        /tikz/every even column/.append style={column sep=3pt}
    },
    legend columns=2,
    colormap name=bright,
    every axis/.append style={font=\footnotesize},
    xmajorgrids,
    ymajorgrids,
    ylabel near ticks,
    clip=false,
]

\def\grouponecolor{GreenColor}
\def\grouptwocolor{RedColor}

\addplot+[color = \grouponecolor, thick, mark=triangle, mark repeat=1]
table[x=partition,y=4]{\ExactLatencyScaleHundredPartitionAblation};
\addlegendentry{q4};
\addplot+[color = \grouponecolor, thick, mark=triangle*, mark options={color=\grouponecolor}, mark repeat=1]
table[x=partition,y=19]{\ExactLatencyScaleHundredPartitionAblation};
\addlegendentry{q19};
\addplot+[color = \grouponecolor, thick, mark=triangle, dashed, mark options={solid}, mark repeat=1]
table[x=partition,y=21]{\ExactLatencyScaleHundredPartitionAblation};
\addlegendentry{q21};

\addplot+[color = \grouptwocolor, thick, mark=diamond, mark repeat=1]
table[x=partition,y=13]{\ExactLatencyScaleHundredPartitionAblation};
\addlegendentry{q13};
\addplot+[color = \grouptwocolor, thick, mark=diamond*, mark options={color=\grouptwocolor}, mark repeat=1]
table[x=partition,y=15]{\ExactLatencyScaleHundredPartitionAblation};
\addlegendentry{q15};
\addplot+[color = \grouptwocolor, thick, mark=diamond, mark repeat=1]
table[x=partition,y=22]{\ExactLatencyScaleHundredPartitionAblation};
\addlegendentry{q22};




\end{axis}
\end{tikzpicture}
\vspace{-4mm}
\caption{Impact of partition size
    (red queries---q13, q15, q22--incur higher merge costs
        than green ones---q4, q19, q21)
}
\label{fig:ablation-partition}
\vspace{-3mm}
\end{figure}
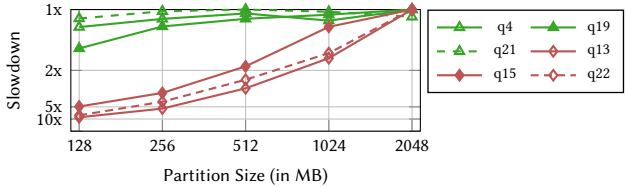


\ignore{
\subsection{Pipelined Query Processing}
\label{sec:ablation-timeline}

A key factor for \system's total query latency to be comparable or lower than other exact processing systems despite the additional overhead of merge computations is pipelining of different operations. 
Since each operation, given its input \mdf, can operate independently, a pipelined execution leads to higher utilization of available compute, reducing the total time taken.
In Figure~\ref{fig:timeline-view}, we plot how different nodes process data partitions over the execution of a query.
Although the query executes for a larger duration, we limit the observed timeline to 1 second for the sake of clarity. The operation \textsf{agg} is the final node for query Q6. Thus, each block in the timeline of operation \textsf{agg} represents an obtained output.
}



\section{Related Work}
\label{sec:related}

Approximate Query Processing (AQP) allows users to analyze large datasets interactively 
    at a fraction of the costs of executing exact queries. 
Despite the years of research in AQP, 
    it has been less successful in supporting deeply nested queries~\cite{chaudhri2017nosilverbullet}.




\paragraph{Online Aggregation} 

Hellerstein et al.~\cite{hellerstein1997online} first introduced the idea of OLA.
Since then, various works~\cite{haas1999ripple, luo2002scalable, dittrich2002progressive, jermaine2006sort, li2016wander} have built on top to increase the extent of queries supported, more focused on join queries. 
RippleJoin~\cite{haas1999ripple} progressively joins multiple tables, 
    \cite{luo2002scalable} improves online join algorithms for queries with low cardinality or a high number of groups. 
\cite{dittrich2002progressive} provides OLA support for sort-based join algorithms. 
\cite{jermaine2006sort, jermaine2008scalable} 
    provides a scalable disk-based approach while providing better statistical bounds for sort-based join algorithms. 
WanderJoin~\cite{li2016wander} efficiently handles queries with multiple joins using indexes.
Supporting nested aggregates and sub-queries has been limited in the literature. 
G-OLA~\cite{zeng2015g} generalizes OLA to nested predicates 
    by dividing the processed tuples into certain and uncertain sets based on running estimates. 
Since its efficiency relies on the quality of the estimates, 
    an incorrect estimation can lead to re-computation.
Some related works~\cite{wu2009distributed} also look at OLA in a distributed setting. 
\cite{condie2010online, pansare2011online, shi2012cola} further applies OLA to MapReduce jobs, 
    but are limited to a basic template of queries.

\paragraph{Incremental View Maintenance} 

Materialized views~\cite{blakeley1986efficiently, chaudhuri1995optimizing,agrawal1997efficient} provide improved query performance at the cost of additional storage and maintenance overhead. 
The base table changes require updating the materialized view using a delta query. 
Various techniques~\cite{gupta1993maintaining, yang1997algorithms, salem2000roll, goldstein2001optimizing, nikolic2014linview, ahmad2009dbtoaster, billy2023icde} 
    have been proposed to perform incremental updates to materialized views.
View maintenance and indexes~\cite{chockchowwat2022automatically, chockchowwat2022airphant, park2015neighbor} are fundamentally different from OLA
    since unlike OLA, it does not aim to estimate future results,
which involves managing/propagating uncertainty over multiple operations.

\paragraph{Approximate Query Processing} 

AQP includes synopses-based techniques~\cite{cormode2011synopses} using samples~\cite{chaudhuri2007optimized, babcock2003dynamic, agarwal2013blinkdb, park2018verdictdb, park2017database, bater2020saqe, park2019blinkml, he2018demonstration, park2016visualization}, wavelets~\cite{chakrabarti2001approximate}, histograms~\cite{ioannidis1999histogram, poosala1999approximate}, sketches~\cite{cormode2011sketch}, etc. 
Recent ML-based works formulate AQP as a data-learning problem---through 
    density estimators~\cite{hilprecht2019deepdb, ma2019dbest}, regression models~\cite{ma2019dbest}, 
    generative models~\cite{park2018tablegan, xu2019ctgan, thirumuruganathan2020approximate}.
More recently, query-aware generative models have been used to improve approximation error 
    for low-cardinality queries~\cite{zhang2021approximate, sheoran2022electra}. 
Challenges with these approaches include the high cost of model maintenance and re-training, the limited set of supported queries, and the difficulty in providing correctness bounds.

\paragraph{Cardinality Estimation}

Cardinality estimation is a fundamental problem in query optimization. Traditional approaches include using synopses like histograms, sketches, and samples~\cite{cormode2011synopses, park2020quicksel}. Recently, learning-based methods involving deep autoregressive models~\cite{yang2019deep, yang2020neurocard}, and ensemble-based methods~\cite{liu2021fauce} are gaining traction. 
\cite{wang2020we} provides a thorough comparison.
Any improvements in cardinality estimation can also improve \system's accuracy.







\section{Conclusion}
\label{sec:conclusion}

In this paper, we take a step toward Deep OLA (Online Aggregation), by introducing a novel data model that is \emph{closed} under set-oriented operations (e.g., map/filter/join/agg), thus enabling 
    applications of nested operations
        to previous OLA outputs. 
We show its viability through \system---a Deep OLA system implemented in Rust. 
That is, we have evaluated \system on TPC-H (100 GB) by comparing 
    it against state-of-the-art OLA engines and conventional data systems. 
Our experiments show that \system provides first estimates 
    \mediantimebest faster (median) than conventional systems' computing exact answers
    while offering a median \medianfirsterror relative error.
Moreover, \system incurs small overhead (\medianoverheadbest median slowdown) in
    producing exact answers.
In fact, the pipelined implementation of different ops in \system 
    often provides faster total latencies 
        than exact query processing engines for some queries.
In the future, we aim to extend \system to support a SQL-like declarative interface 
    with automated query optimizations. 
We also aim to extend \system to a distributed setup, 
    making Deep OLA further scalable.




\begin{acks}
This work is supported in part by Microsoft Azure.
\end{acks}

\clearpage

\bibliographystyle{ACM-Reference-Format}
\bibliography{acmart}

\appendix
\section{Inference Correctness Proofs}
\label{sec:appendix-inference}
The following are restated lemmas and correctness proofs of the aggregate inference (\cref{sec:inference}). Given observations $(y_{i,:t}, x_{i,:t})$ up until current progress $t$, \cref{lemma:unbiased-count} and \cref{lemma:unbiased-agg} together show that \system's aggregate inference is unbiased under some conditions.

\lemmaunbiasedcount*
\begin{proof}
    To justify A), $w$ is unbiased if a sufficient condition holds: $\overline{X}_t = \frac{1}{m} \sum_{i=1}^{m(t)} X_i(t) = c_i t^w e^\eps$ where $e^\eps$ is a multiplicative error with $\E[\eps] = 0$ and independence to observations. Notice that $\log \overline{X}_t = \log c_i + w \log t + \eps$. By an unbiasedness of ordinary least square estimator, the estimated $w$ coincides with the expectation of $w$.

    To justify B), we show that all operations in \cref{sec:model} can satisfy the property possibly with some additional conditions depending on the operation. First, a data source produces monomial rows by implementation since it computes progress $t$ from the number of rows read divided by the final number of rows. Mapping trivially preserves the cardinality. Joins and filters preserve monomial relationship if each tuple has uniform selectivity and the selectivity is independent from the progress. Lastly, aggregation preserves monomial relationship if the groupby key is either uniformly distributed across progress or is a clustered key.
    
    If B) holds, $\E[X_i(t)] = c_i t^w$ for all \mdf including the input to the aggregation. Because \system has already observed $x_{i, t}$ at $t$, $\E[X_i(t) | x_{i,:t}] = \E[X_i(t) | x_{i,t}] = x_{i, t}$, implying $c_i = x_{i,t} / t^w$ and leading to \system's group cardinality estimator when evaluated at $t = T = 1$.
\end{proof}

\lemmaunbiasedagg*
\begin{proof}
    This is true for $f_{count}$ because $\xh_{i, :t}$ is given.
    
    Assuming the terms $\{U_{i,j}\}_{j=1}^{X_i(t)}$ in summation are stationary with a constant $\E[U_{i,j}] = \E[U_i]$, $f_{sum}$ is unbiased. Let $Y_i(t) = \sum_{j=1}^{X_i(t)} U_{i,j}$, we can use stationarity to factor $\E[Y_i(T) \; | \; x_{i,:t}, y_{i,:t}] = \E[X_i(T) \; | \; x_{i,:t}] E[U_i \; | \; x_{i,:t}, y_{i,:t}]$. By expanding the conditional expectation, $E[U_i \; | \; x_{i,:t}, y_{i,:t}] = \frac{y_{i,t}}{x_{i,t}}$ which equates the final result to $f_{sum}$.
    
    Assuming that elements and weights are independent, weighted average estimate $f_{avg}$ is unbiased by expectation linearity.
    
    Assuming equal-frequency assumption (i.e. samples are uniformly distributed into $Y$ distinct values), $f_{cd}$ is an unbiased estimator as described in~\cite{Haas1995Distinct}. Additionally, the same work also derives an estimator for skewed frequency, left for future improvements and experiments.

    The $q$-th sample quantile is asymptotically normal around the $q$-th population quantile with variance $q(1-q) / (X_i(t) p_q^2)$ where $p_q$ is the density of the quantile value studied in section 21.2 of~\cite{Vaart1998AsympStats}. Assuming sufficiently large sample size $x_{i,t} = \omega(1/p_q^2)$, which is typical in OLA due to its data size, the sample quantile is unbiased with small variance.
\end{proof}

\section{List of Variance Propagation}
\label{sec:appendix-ci}

\paragraph{Count and Cardinality Estimation} Refer to $f_{\text{count}}$ and \cref{eq:cardinality-estimate}:
\begin{equation}
    \Var(f_{\text{count}}) = \Var(\xh_{i, :t}) = (\xh_{i, :t} \ln(1/t))^2 \Var(w)
\end{equation}

\paragraph{Sum} Recall $f_{\text{sum}} = \frac{y_{i,t}}{x_{i,t}} \xh_{i, :t}$. If summation over mutable attributes, apply $\Var(A + B) = \Var(A) + \Var(B) + 2 \Cov(A, B)$. Otherwise, if summation over constant attributes:
\begin{equation} \label{eq:sum-var}
    \Var(f_{\text{sum}}) = \frac{1}{x_{i,t}^2} \left[ \Var(y_{i,t}) \xh_{i, :t}^2 + \Var(\xh_{i, :t}) y_{i,t}^2 \right]
\end{equation}

\paragraph{Weighted Average} Refers to \cref{eq:scale-weighted-avg}). This is consistent with applying \cref{eq:sum-var} on both the weighted summation and the summation of weights when accounting for their covariance:
\begin{equation}
    \Var(f_{\text{avg}}) = \left[ \frac{y'_{i,t}}{y''_{i,t}} \right]^2 \left[ \frac{\Var(y'_{i,t} / x_{i,t})}{(y'_{i,t} / x_{i,t})^2} + \frac{\Var(y''_{i,t} / x_{i,t})}{(y''_{i,t} / x_{i,t})^2} \right]
\end{equation}

\paragraph{Count Distinct} Refers to \cref{eq:mm1-main}. The variance propagation rule for count distinct can be written in terms of $\frac{\partial Y}{\partial y_{i,t}}$ and $\frac{\partial Y}{\partial \xh_{i, :t}}$ similarly based on \cref{eq:propagation-of-uncertainty}. Both of these terms can be derived through implicit differentiation on \cref{eq:mm1-main}. Let $h'(z) = \frac{dh}{dz}$ derived in terms of digamma functions (calculated in logarithmic terms for numerical stability). The two implicit differentiation equations are the following.
\begin{equation}
    1 = \frac{\partial Y}{\partial y_{i,t}} ( 1 - h(\xh_{i, :t} / Y) ) + \frac{\partial Y}{\partial y_{i,t}} \frac{\xh_{i, :t}}{Y} h'(\xh_{i, :t} / Y)
\end{equation}
\begin{equation}
    0 = \frac{\partial Y}{\partial \xh_{i, :t}} ( 1 - h(\xh_{i, :t} / Y) ) + \left( \frac{\partial Y}{\partial \xh_{i, :t}} \frac{\xh_{i, :t}}{Y} - 1 \right) h'(\xh_{i, :t} / Y)
\end{equation}

Solving them gives us the desired derivatives.
\begin{equation}
    \frac{\partial Y}{\partial y_{i,t}} = \frac{1}{( 1 - h(\xh_{i, :t} / Y) ) + (\xh_{i, :t} / Y) h'(\xh_{i, :t} / Y)}
\end{equation}
\begin{equation}
    \frac{\partial Y}{\partial \xh_{i, :t}} = \frac{h'(\xh_{i, :t} / Y)}{( 1 - h(\xh_{i, :t} / Y) ) + (\xh_{i, :t} / Y h')(\xh_{i, :t} / Y)}
\end{equation}

Lastly, apply \cref{eq:propagation-of-uncertainty}.
\begin{equation}
    \Var(f_{\text{cd}}) = \Var(Y) = \frac{\Var(y_{i,t}) + \Var(\xh_{i, :t}) [h'(\xh_{i, :t} / Y)]^2}{[( 1 - h(\xh_{i, :t} / Y) ) + (\xh_{i, :t} / Y h')(\xh_{i, :t} / Y)]^2}
\end{equation}

\paragraph{Order Statistics} Outputs estimated variance $\Var(y_{i,t})$.

\paragraph{Mapping and Projection} Relies on automatic differentiation to evaluate partial derivatives with respect to all mutable attributes in the mapping. This excludes those mapping functions with undefined derivative on the current attribute value (e.g., $|x|$ at $x = 0$, or a step function), in which \system can mark the confidence interval as ``unstable.'' Projection is a special case of mapping, in which the projected variance is equal to the source variance.

\section{Pipelined Query Processing}
\label{sec:ablation-timeline}

\begin{figure}[t]

\begin{tikzpicture}

\tikzset{
mynode/.style={
    fill=GreenColor!50!white,draw=black,anchor=west,minimum height=2mm,
    minimum width=0.2mm,inner xsep=0mm,
},
mylabel/.style={
    font=\footnotesize\sf,anchor=east,
},
myxtick/.style={
    font=\scriptsize\sf,anchor=north
},
}

\node[draw=black,minimum height=16mm,minimum width=65mm,anchor=south west] 
    (C) at (0,0) {};

\coordinate (N1) at ($(C.north west)!0.125!(C.south west)$);
\coordinate (N2) at ($(C.north west)!0.375!(C.south west)$);
\coordinate (N3) at ($(C.north west)!0.625!(C.south west)$);
\coordinate (N4) at ($(C.north west)!0.875!(C.south west)$);

\foreach \x in {1, 2, 3, 4, 5, 6, 7, 8, 9} {
    \def \r {0.1 * \x}
    \coordinate (X1\x) at ($(C.north west)!\r!(C.north east)$);
    \coordinate (X2\x) at ($(C.south west)!\r!(C.south east)$);
}


\coordinate (N1E) at ($(C.north east)!(N1)!(C.south east)$);
\coordinate (N2E) at ($(C.north east)!(N2)!(C.south east)$);
\coordinate (N3E) at ($(C.north east)!(N3)!(C.south east)$);
\coordinate (N4E) at ($(C.north east)!(N4)!(C.south east)$);


\node[mylabel] at (N1) {read(lineitem)};
\node[mylabel] at (N2) {filter};
\node[mylabel] at (N3) {map};
\node[mylabel] at (N4) {agg};

\foreach \x in {1, 2, 3, 4, 5, 6, 7, 8, 9} {
    \draw[draw=gray] (X1\x) -- (X2\x);
    \pgfmathtruncatemacro \tms {100 * \x}
    \node[myxtick] at ($(X2\x)+(0,0)$) {\tms ms};
}

\node[myxtick] at (C.south west) {0ms};
\node[myxtick] at (C.south east) {1sec};


\node[mynode,fit={ (N1) ($(N1)!0.092!(N1E)$) }] {};
\node[mynode,fit={ ($(N1)!0.092!(N1E)$) ($(N1)!0.167!(N1E)$) }] {};
\node[mynode,fit={ ($(N1)!0.167!(N1E)$) ($(N1)!0.264!(N1E)$) }] {};
\node[mynode,fit={ ($(N1)!0.264!(N1E)$) ($(N1)!0.334!(N1E)$) }] {};
\node[mynode,fit={ ($(N1)!0.334!(N1E)$) ($(N1)!0.405!(N1E)$) }] {};
\node[mynode,fit={ ($(N1)!0.405!(N1E)$) ($(N1)!0.476!(N1E)$) }] {};
\node[mynode,fit={ ($(N1)!0.476!(N1E)$) ($(N1)!0.547!(N1E)$) }] {};
\node[mynode,fit={ ($(N1)!0.547!(N1E)$) ($(N1)!0.618!(N1E)$) }] {};
\node[mynode,fit={ ($(N1)!0.618!(N1E)$) ($(N1)!0.776!(N1E)$) }] {};
\node[mynode,fit={ ($(N1)!0.776!(N1E)$) ($(N1)!0.847!(N1E)$) }] {};
\node[mynode,fit={ ($(N1)!0.847!(N1E)$) ($(N1)!0.942!(N1E)$) }] {};
\node[mynode,fit={ ($(N1)!0.942!(N1E)$) ($(N1)!1!(N1E)$) }] {};

\node[mynode,fit={ ($(N2)!0.092!(N2E)$) ($(N2)!0.120!(N2E)$) }] {};
\node[mynode,fit={ ($(N2)!0.167!(N2E)$) ($(N2)!0.192!(N2E)$) }] {};
\node[mynode,fit={ ($(N2)!0.264!(N2E)$) ($(N2)!0.278!(N2E)$) }] {};
\node[mynode,fit={ ($(N2)!0.334!(N2E)$) ($(N2)!0.359!(N2E)$) }] {};
\node[mynode,fit={ ($(N2)!0.405!(N2E)$) ($(N2)!0.441!(N2E)$) }] {};
\node[mynode,fit={ ($(N2)!0.476!(N2E)$) ($(N2)!0.498!(N2E)$) }] {};
\node[mynode,fit={ ($(N2)!0.547!(N2E)$) ($(N2)!0.570!(N2E)$) }] {};
\node[mynode,fit={ ($(N2)!0.619!(N2E)$) ($(N2)!0.641!(N2E)$) }] {};
\node[mynode,fit={ ($(N2)!0.771!(N2E)$) ($(N2)!0.792!(N2E)$) }] {};
\node[mynode,fit={ ($(N2)!0.847!(N2E)$) ($(N2)!0.869!(N2E)$) }] {};
\node[mynode,fit={ ($(N2)!0.943!(N2E)$) ($(N2)!0.963!(N2E)$) }] {};

\node[mynode,fit={ ($(N3)!0.120!(N3E)$) ($(N3)!0.121!(N3E)$) }] {};
\node[mynode,fit={ ($(N3)!0.191!(N3E)$) ($(N3)!0.192!(N3E)$) }] {};
\node[mynode,fit={ ($(N3)!0.2786!(N3E)$) ($(N3)!0.2788!(N3E)$) }] {};
\node[mynode,fit={ ($(N3)!0.3592!(N3E)$) ($(N3)!0.3594!(N3E)$) }] {};
\node[mynode,fit={ ($(N3)!0.4408!(N3E)$) ($(N3)!0.441!(N3E)$) }] {};
\node[mynode,fit={ ($(N3)!0.4983!(N3E)$) ($(N3)!0.4985!(N3E)$) }] {};
\node[mynode,fit={ ($(N3)!0.5701!(N3E)$) ($(N3)!0.5702!(N3E)$) }] {};
\node[mynode,fit={ ($(N3)!0.6413!(N3E)$) ($(N3)!0.6415!(N3E)$) }] {};
\node[mynode,fit={ ($(N3)!0.7925!(N3E)$) ($(N3)!0.7926!(N3E)$) }] {};
\node[mynode,fit={ ($(N3)!0.8692!(N3E)$) ($(N3)!0.8694!(N3E)$) }] {};
\node[mynode,fit={ ($(N3)!0.9628!(N3E)$) ($(N3)!0.9630!(N3E)$) }] {};

\node[mynode,fit={ ($(N4)!0.121!(N4E)$) ($(N4)!0.132!(N4E)$) }] {};
\node[mynode,fit={ ($(N4)!0.192!(N4E)$) ($(N4)!0.194!(N4E)$) }] {};
\node[mynode,fit={ ($(N4)!0.278!(N4E)$) ($(N4)!0.280!(N4E)$) }] {};
\node[mynode,fit={ ($(N4)!0.359!(N4E)$) ($(N4)!0.361!(N4E)$) }] {};
\node[mynode,fit={ ($(N4)!0.441!(N4E)$) ($(N4)!0.443!(N4E)$) }] {};
\node[mynode,fit={ ($(N4)!0.499!(N4E)$) ($(N4)!0.500!(N4E)$) }] {};
\node[mynode,fit={ ($(N4)!0.570!(N4E)$) ($(N4)!0.572!(N4E)$) }] {};
\node[mynode,fit={ ($(N4)!0.641!(N4E)$) ($(N4)!0.643!(N4E)$) }] {};
\node[mynode,fit={ ($(N4)!0.792!(N4E)$) ($(N4)!0.794!(N4E)$) }] {};
\node[mynode,fit={ ($(N4)!0.869!(N4E)$) ($(N4)!0.871!(N4E)$) }] {};
\node[mynode,fit={ ($(N4)!0.963!(N4E)$) ($(N4)!0.965!(N4E)$) }] {};

\end{tikzpicture}

\caption{Pipelined execution of example query Q6}
\label{fig:timeline-view}
\end{figure}
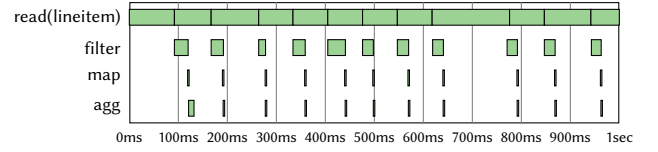

A key factor for \system's total query latency to be comparable or lower than other exact processing systems despite the additional overhead of merge computations is pipelining of different operations. 
Since each operation, given its input \mdf, can operate independently, a pipelined execution leads to higher utilization of available compute, reducing the total time taken.
In Figure~\ref{fig:timeline-view}, we plot how different nodes process data partitions over the execution of a query.
Although the query executes for a larger duration, we limit the observed timeline to 1 second for the sake of clarity. The operation \textsf{agg} is the final node for query Q6. Thus, each block in the timeline of operation \textsf{agg} represents an obtained output.

\end{document}